\documentclass[ba,linksfromyear,nameyear]{imsart}
\usepackage{algorithm,algorithmic}
\usepackage{amsmath,amsthm,amsfonts}
\usepackage{natbib}
\usepackage{graphicx}
\usepackage{hyperref}
\usepackage{epstopdf}

\startlocaldefs
\theoremstyle{plain}
\newtheorem{thm}{Theorem}
\newtheorem{lem}{Lemma}
\endlocaldefs
\pubyear{2016}
\volume{TBA}
\issue{TBA}
\firstpage{1}
\lastpage{29}
\doi{10.1214/16-BA1017}

\begin{document}

\begin{frontmatter}
\title{Bayesian Solution Uncertainty Quantification for Differential Equations\thanksref{T1}}
%
\relateddois{T1}{Related articles:
DOI:~\relateddoi[ms=BAXXX]{Related item:}{10.1214/16-BAXXX},
DOI:~\relateddoi[ms=BAXXX]{Related item:}{10.1214/16-BAXXX},
DOI:~\relateddoi[ms=BAXXX]{Related item:}{10.1214/16-BAXXX};
rejoinder at
DOI:~\relateddoi[ms=BAXXX]{Related item:}{10.1214/16-BAXXX}.}
%
%
\runtitle{Bayesian Solution Uncertainty Quantification for
Differential Equations}

\begin{aug}
\author[addr1]{\fnms{Oksana A.} \snm{Chkrebtii}\ead[label=e1]{oksana@stat.osu.edu}},
\author[addr2]{\fnms{David A.} \snm{Campbell}\ead[label=e2]{dac5@stat.sfu.ca}},
\author[addr3]{\fnms{Ben} \snm{Calderhead}\ead[label=e3]{b.calderhead@ucl.ac.uk}},
\and
\author[addr4]{\fnms{Mark A.} \snm{Girolami}\ead[label=e4]{M.Girolami@warwick.ac.uk}}

\runauthor{O. A. Chkrebtii, D. A. Campbell, B. Calderhead, and M. A. Girolami}

\address[addr1]{Department of Statistics, The Ohio State University, Columbus, OH, USA, \printead{e1}}
\address[addr2]{Department of Statistics \& Actuarial Science, Simon Fraser University, Canada, \printead{e2}}
\address[addr3]{Department of Mathematics, Imperial College London, London, UK, \printead{e3}}
\address[addr4]{Department of Statistics, The University of Warwick, Coventry, UK, \printead{e4}}
\end{aug}

%
\begin{abstract}
We explore probability modelling of discretization uncertainty for
system states defined implicitly by ordinary or partial differential
equations. Accounting for this uncertainty can avoid posterior
under-coverage when likelihoods are constructed from a coarsely
discretized approximation to system equations. A formalism is proposed
for inferring a fixed but \textit{a priori} unknown model trajectory
through Bayesian updating of a prior process conditional on model
information. A one-step-ahead sampling scheme for interrogating the
model is described, its consistency and first order convergence
properties are proved, and its computational complexity is shown to be
proportional to that of numerical explicit one-step solvers. Examples
illustrate the flexibility of this framework to deal with a wide
variety of complex and large-scale systems. Within the calibration
problem, discretization uncertainty defines a layer in the Bayesian
hierarchy, and a Markov chain Monte Carlo algorithm that targets this
posterior distribution is presented. This formalism is used for
inference on the JAK-STAT delay differential equation model of protein
dynamics from indirectly observed measurements. The discussion outlines
implications for the new field of probabilistic numerics.
\end{abstract}

%

\end{frontmatter}


\section{Quantifying uncertainty for differential equation models}\label{sec:PoseriorUQ}

Many scientific, economic, and engineering disciplines represent the
spatio-temporal evolution of complex systems implicitly as differential
equations, using few but readily interpretable parameters. Differential
equation models describe the natural dependence between system states
and their rates of change on an open spatio-temporal domain, $\mathcal
{D}\subset\mathbb{R}^{d}$. Their relationship, represented by the
function $F$, is fully specified up to some physical constants $\theta
$ belonging to a parameter space $\Theta$. Mathematically, system
states, $u: \mathcal{D}\times\Theta\to\mathbb{R}^p$ satisfy the
partial differential equation (PDE),
\[
F\left(x, t,\frac{\partial u}{\partial t},\frac{\partial u}{\partial
x_1},\frac{\partial u}{\partial x_2},\frac{\partial^2 u}{\partial
x_1\partial x_2},\ldots, u, \theta\right) = 0, \quad (x,t)\in
\mathcal{D}\cup\partial\mathcal{D},
\]
and constraints at the boundary $\partial\mathcal{D}$ of $\mathcal{D}$.
When states evolve with respect to a single variable, such as time, the
model simplifies to an ordinary differential equation (ODE). The
explicit solution, denoted $u^*: \mathcal{D}\times\Theta\to\mathbb
{R}^p$, of a differential equation problem over spatio-temporal
locations, $(x,t)\in\mathcal{D}$, is used to explore system dynamics,
design experiments, or extrapolate to expensive or dangerous
circumstances. Therefore, increasing attention is being paid to the
challenges associated with quantifying uncertainty for systems defined
by such models, and in particular those based on mathematical
descriptions of complex phenomena such as the weather, ocean currents,
ice sheet flow, and cellular protein transport \citep
{GhanemSpanos2003,KaipioEtAl2004,HuttunenKaipio2007,Marzouk2007,Marzouk2009,Stuart2010}.

The main challenge of working with differential equation models, from
both mathematical and statistical perspectives, is that solutions are
generally not available in closed form. When this occurs, prior
exploration of the mathematical model, $u^*(x, t, \theta)$, cannot be
performed directly. This issue is dealt with throughout the literature
by replacing the exact solution with an $N$ dimensional approximate
solution, $\hat{u}^N(x, t,\theta)$, obtained using numerical
techniques over a size $N$ grid partitioning of the domain $\mathcal
{D}$. Inference and prediction based on $\hat{u}^N(x, t,\theta)$ then
inform the reasoning process when, for example, assessing financial
risk in deciding on oil field bore configurations, or forming
government policy in response to extreme weather events.

Limited computation and coarse mesh size are contributors to numerical
error which, if assumed negligible can lead to serious misspecification
for these highly nonlinear models. The study of how uncertainty
propagates through a mathematical model is known as the \textit
{forward problem}.
Numerical error analysis provides local and global discretization error
bounds that are characterized point-wise and relate to the asymptotic
behaviour of the deterministic approximation of the model. However,
accounting for this type of {\em verification error} for the purpose of
model inference has proven to be a difficult open problem \citep
{Obercampf2010}, causing discretization uncertainty to be often ignored
in practice. For complex, large-scale models, maintaining a mesh size
that is fine enough to ignore discretization uncertainty demands
prohibitive computational cost \citep{ArridgeEtAl2006}, so that a
coarse mesh is the only feasible choice despite associated model
uncertainty \citep{Agora2014}. Furthermore, some ODE and PDE models
exhibit highly structured but seemingly unpredictable behaviour.
Although the long term behaviour of such systems is entrenched in the
initial states, the presence of numerical discretization error
introduces small perturbations which eventually lead to exponential
divergence from the exact solution. Consequently, information about
initial states rapidly decays as system solution evolves in time \citep
{Berliner1991}.

The {\it calibration} or {\it statistical inverse problem} concerns the
uncertainty in unknown model inputs, $\theta$, given measurements of
the model states
\citep{Bock1983,IonidesEtAl2006, Dowd2007,RamsayEtAl2007,
XueEtAl2010,Brunel2008,
LiangWu2008,calderhead2011statistical,CampbellSteele2011,
GugushviliKlaassen2012,CampbellChkrebtii2013, XunEtAl2014}. However,
inference for differential equation models lacking a closed form
solution has mainly proceeded under the unverified assumption that
numerical error is negligible over the entire parameter space. From the
perspective of \cite{KennedyOhagan2001}, the discrepancy between the
numerical approximation, $\hat{u}^N$, and the unobservable
mathematical model, ${u}^*$, at the observation locations $({\bf
x},{\bf t})\in\mathcal{D}^T$ would be modelled by the function,
$\delta({\bf x}, {\bf t})= {u}^*({\bf x},{\bf t}, \theta)-\hat
{u}^N({\bf x},{\bf t}, \theta)$. Estimation of $\delta({\bf x},{\bf
t})$ requires observations $y({\bf x},{\bf t})$ from $ u^*({\bf x},{\bf
t}, \theta)$. However, magnitude and structure of discretization error
change with $\theta$, while $\delta({\bf x},{\bf t})$ is agnostic to
the specific form of the underlying mathematical model, treating it as
a ``black box''. Consequently exploring the model uncertainty for given
values of $\theta$ requires observations for each parameter setting or
the assumption that $\delta({\bf x}, {\bf t})$ is independent of
$\theta$. Alternatively, in the discussion of \cite
{KennedyOhagan2001}, H. Wynn argues that the sensitivity equations and
other underlying mathematical structures that govern the numerical
approximation could be included in the overall uncertainty analysis.

\subsection{Contributions and organization of the paper}

This paper develops a Bayesian formalism for modelling discretization
uncertainty within the forward problem, and explores its contribution
to posterior uncertainty in the model parameters.
We substantially expand and formalize ideas first described by \citet
{Skilling1991} and develop a general class of probability models for
solution uncertainty for differential equations in the so-called
explicit form,
\begin{align}
\mathcal{A}=f(x,t,\mathcal{B},\theta), \quad\; (x,t)\in\mathcal
{D}\cup\partial\mathcal{D}, \label{eqn:vecfield}
\end{align}
where $\mathcal{A}$ and $\mathcal{B}$ are subsets of the
deterministic solution $u=u(x,t, \theta)$ and its partial derivatives,
and the function $f$ is Lipschitz continuous in $\mathcal{B}$.
Although (\ref{eqn:vecfield}) is very general, Sections \ref
{sec:priors} and \ref{sec:probalgorithm} consider ordinary
differential equations for expositional simplicity, where $\mathcal
{A}=du(t, \theta)/dt:=u_t$ and $\mathcal{B}=u(t, \theta):=u$ for
$t\in[0, L], L>0$,
i.e. $u_t = f(t,u,\theta)$, with a given initial condition. Sections
\ref{sec:Inverse} and \ref{sec:FsimAndInference} extend this to the
generality of (\ref{eqn:vecfield}), showcasing a wide class of ODEs
and PDEs, such as mixed boundary value problems with multiple
solutions, delay initial function problems, and chaotic differential equations.

Section \ref{sec:priors} considers fixed model parameters $\theta$,
hyperparameters $\Psi$, and discretization grid of size $N$. Prior
uncertainty about the solution $u$ and its derivative $u_t$ is encoded
in a probability model $[u , u_t \mid\theta,\Psi]$ on a space of
functions satisfying the initial condition. Section \ref
{sec:probalgorithm} describes a sequential model interrogation method
that consists of (i) drawing a sample $u^{n-1}$ from the marginal $[u
\mid \mbox{f}_1,\ldots, \mbox{f}_{n-1},\theta,\Psi]$, and (ii)
using this to compute \textit{model interrogations} $\mbox{f}_n
\equiv f(t_n, u^{n-1},\theta)$ at grid locations $t_n\in[0, L], L>0$.
Conditioning on these incrementally obtained interrogations defines the
sequence of models $[u , u_t\mid \mbox{f}_1,\theta,\Psi],\ldots,
[u , u_t\mid \mbox{f}_1,\ldots, \mbox{f}_N,\theta,\Psi]$. In
Section \ref{sec:rate} we marginalize over ${\bf f} = (\mbox
{f}_1,\ldots,\mbox{f}_N)$ to obtain the probability model, or \textit
{probabilistic solution},
%
\begin{equation}
\left[u \mid\theta, \Psi, N\right] = \smallint\smallint\left
[u,u_t , {\bf f} \mid\theta, \Psi, N\right] \, \mbox{d}u_t \,
\mbox{d}{\bf f},
\label{eqn:posterior_solution}
\end{equation}
which contracts to the true deterministic solution $u$ as $N\to\infty
$. Computational complexity is shown in Section \ref{sec:complexity}
to be of the same order as comparably sampled explicit numerical
solvers, and scalability is illustrated in Section \ref
{sec:FsimAndInference} for the chaotic Navier--Stokes system involving
the solution of over 16,000 coupled stiff ordinary differential
equations. The probabilistic\vadjust{\eject} solution defines a trade-off between
discretization uncertainty and computational cost. To illustrate this
point, Figure \ref{fig:sinusoid100} compares the exact, but \textit{a
priori} unknown solution with sample trajectories from the proposed
probability model (\ref{eqn:posterior_solution}) and a first order
numerical scheme obtained over progressively refined discretization
grids. The exact solution lies within a high posterior density region
of (\ref{eqn:posterior_solution}), and as the grid is refined yielding
more information about the model, knowledge of the solution increases,
concentrating to the exact solution. Strong posterior non-stationarity
of the error structure (\ref{eqn:posterior_solution}) is illustrated
in Figure \ref{fig:lorenzerrors} for the three states of a system
whose trajectory is restricted around a chaotic attractor (details are
provided in Section \ref{sec:lorenz}).

\begin{figure}
\includegraphics{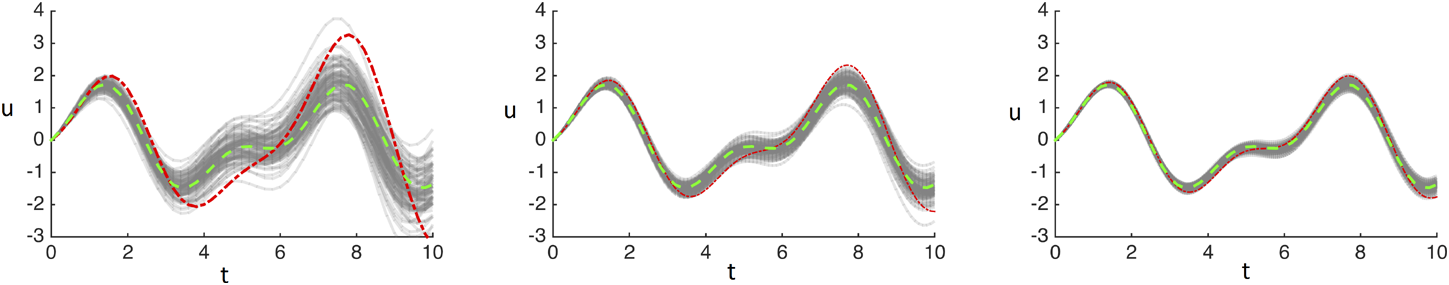}
\caption{Exact solution (green line) of ODE (\ref{eqn:toyIVPsystem});
100 draws (grey lines) from $p\{u({\bf t}) \mid\Psi, N \}$
over a dense grid ${\bf t}$; first order explicit Euler approximation
(red line), obtained using equally-spaced domain partitions of sizes
$N=50, 100, 200$ (left to right).}\label{fig:sinusoid100}
\end{figure}

\begin{figure}[t!]
\includegraphics[scale=0.58]{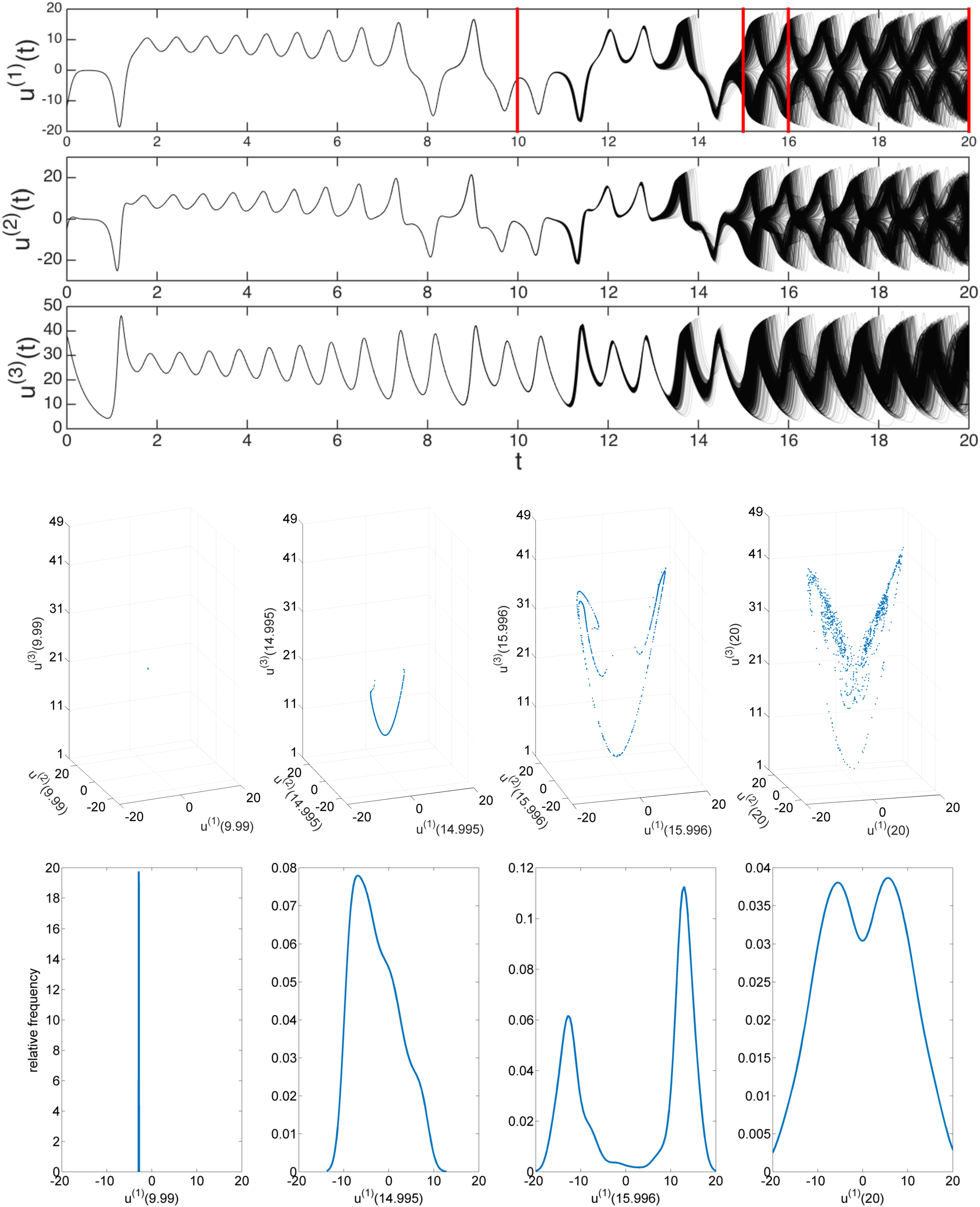}\vspace*{-3pt}
\caption{1000 draws from (\ref{eqn:posterior_solution}) for the
Lorenz63 model with fixed initial states and model parameters in the chaotic regime. Red
vertical lines in $u^{(1)}(t)$ correspond to four time points used in the diagrams in
the last two rows.}
\label{fig:lorenzerrors}\vspace*{-3pt}
\end{figure}

In Section \ref{sec:Inverse} we embed discretization uncertainty in
the inferential hierarchy,
\begin{align}
[ \theta, u, \Psi, \Sigma\mid\mathbf{y}, N ] \; \propto\;
\; \underbrace{[ \mathbf{y} \mid u , \theta, \Sigma] }_{\mbox
{{\footnotesize likelihood}}}
\; \underbrace{\left[ u \mid\theta, \Psi, N \right]}_{\mbox
{{\footnotesize solution uncertainty}}}
\; \underbrace{\left[ \theta, \Psi, \Sigma \right] }_{\mbox
{{\footnotesize prior}}},
\label{eqn:posterior_probabilistic}
\end{align}
where $\Sigma$ includes parameters of the observation process, such as
the covariance structure of the error terms. We adopt a forward
simulation approach for the middle layer within a Markov chain Monte
Carlo (MCMC) algorithm to sample from this posterior distribution. The
importance of explicitly modelling discretization uncertainty within
the inverse problem is illustrated in Section \ref
{subsec:jakstatexample} by the nonlinear posterior dependence between
solver hyperparameters $\Psi$ and differential equation model
parameters $\theta$. This approach provides a useful exploratory tool
to diagnose the presence of discretization effects. The problem of
posterior under-coverage resulting from ignoring discretization induced
solution uncertainty in (\ref{eqn:posterior_probabilistic}) is
illustrated in Section \ref{sec:HeatPDE}.

Discussion of various extensions and implications for the emerging
field of probabilistic numerics is provided in Section \ref
{sec:discussion}. Algorithmic and theoretical details are described in
the Supplement \citep{supplement1017}.  
MATLAB implementations of all examples are provided at
\url{github.com/ochkrebtii/uqdes}.

\subsection{Notation}

We use of the mathematical framework for defining numerical problems on
function spaces developed through the work of \cite{Stuart2010}, \cite
{KaipioSomersalo2007}, \cite{OHagan1992}, and \cite{Skilling1991}. We
denote partial derivatives using subscripts to indicate the variable
with respect to which a partial derivative is taken. For example
$u_{xx} := \tfrac{\partial^2}{\partial x^2}u$ refers to the second
partial derivative of $u$ with respect to $x$. Superscripts in
parentheses denote elements of a vector, for example, $u^{(k)}$ is the
$k$th element of the $P$-dimensional vector $u$. The methodology
proposed in this paper is based on a sequentially updated probabilistic
solution, consequently we use superscripts without parentheses, e.g.,
$m^k$, to refer to $k$-times updated functions. Superscripts on random
variables, e.g., $u^k$, represent a sample from its $k$-times updated
distribution. For notational simplicity, we omit the dependence of the
differential equation solution $u$, and derivative $u_t$ on $\theta$,
$\Psi$, $N$, and $x$ and/or $t$ when the context makes this dependence clear.

\section{Prior model for a fixed but unknown solution}\label{sec:priors}

This section is concerned with probability modelling of initial beliefs
about the unknown function, $u^*:\mathcal{D}\times\Theta\to\mathbb
{R}^p$, satisfying an ODE initial value problem of the form,
\begin{align}
u_t = f\left(t,u,\theta\right), \quad t \in[0,L], L>0, \label{eqn:odeivp}
\end{align}
with fixed initial value $u(0)=u^*(0)$. We describe how to incorporate
boundary and regularity assumptions into a flexible Gaussian process
(GP) prior model \citep{Rasmussen2006,Stuart2010} for the unknown
solution. Prior beliefs are updated in Section \ref{sec:probalgorithm}
by conditioning on information about the differential equation model
obtained over a size $N$ partition of the domain $\mathcal{D}$
according to a chosen sampling scheme.

\subsection{GP prior on the solution of an ODE initial value problem}\label{sec:ivpprior}

GPs are a flexible class of stochastic processes that can be defined
jointly with their derivatives in a consistent way through the
specification of an appropriate covariance operator \citep
{SolakMurray2003}. We model uncertainty jointly for the time derivative
in (\ref{eqn:odeivp}) and its solution, $(u_t,u)$, using a GP prior
with covariance structure, $(C_t^0, C^0)$, and mean function, $(m_t^0,
m^0)$. The marginal means must satisfy the constraint $\int
_{0}^{t}m_t^0(z)dz = m^0(t)$. The covariance structure is defined as a
flexible convolution of kernels \citep[see,
e.g.,][]{Higdon1998,Higdon2002}. A covariance operator for the
derivative $u_t$ evaluated between locations $t_j$ and $t_k$ is
obtained by convolving a square integrable kernel $R_\lambda: \mathbb
{R} \times\mathbb{R} \to\mathbb{R}$ with itself,
\[
C_t^0(t_j,t_k)= \alpha^{-1}\smallint_{\mathbb{R}}R_\lambda(t_j,
z)R_\lambda(t_k, z) \mathrm{d}z,
\]
parameterized by the length-scale $\lambda>0$ and prior precision
$\alpha> 0$, to be discussed in Section \ref{sec:hyperparms}. An
example is the infinitely differentiable kernel, $R_\lambda(t_j,t_k)
=\break
\exp\{ -{(t_j-t_k)^2}/2\lambda^{2}\}$, which decreases
exponentially with squared distance between $t_j$ and $t_k$ weighted by
the length-scale $\lambda$. Its convolution yields the \textit
{squared exponential covariance} $C_t^0(t_j,t_k)= \sqrt{\pi}\alpha
^{-1}\lambda\exp\{ -{(t_j - t_k)^2/4\lambda^2}\}$. An
example of a kernel with bounded support is $R_\lambda(t_j,t_k) =
\mbox{1}_{(t_j-\lambda,t_j+\lambda)}(t_k)$, which yields the
piecewise linear \textit{uniform covariance} $C_t^0(t_j,t_k)= \{ \mbox
{min}(t_j,t_k) - \mbox{max}(t_j,t_k) +2\lambda\} \; \mbox
{1}_{(0,\infty)}\{ \mbox{min}(t_j,t_k)-\mbox{max}(t_j,t_k) +
2\lambda\}$. Details are provided in Supplement D.4.
A marginal covariance operator, $C^0$, on the state is obtained by
integrating $C_t^0$ with respect to evaluation points $t_j$ and $t_k$
or, equivalently, by convolving the integrated kernel $Q_\lambda
(t_j,t_k) = \int_{0}^{t_j} R_\lambda(z,t_k)\mbox{d}z$ with itself:
\[
C^0(t_j,t_k) = \alpha^{-1}\smallint_{\mathbb{R}}Q_\lambda(t_j,
z)Q_\lambda(t_k, z) \mathrm{d}z.
\]
Defining the covariance over the state by integration ensures that
$C^0(0,0) = 0$, which enforces the boundary condition, $u(0) = m^0(0)$
(see Supplement D.1; 
 an alternative way to enforce
the boundary constraint is to obtain the derivative process $u_t$ by
differentiation of $u$ and condition the joint prior for $(u_t,u)$ on
the exact initial state, $u^*(0)$). This introduces anisotropy over the
state that is consistent with exact knowledge of the solution at the
initial boundary and increasing uncertainty thereafter. The
cross-covariance terms are defined analogously as $\int_{0}^{t_j}
C_t^0(z,t_k) \mbox{d}z = \alpha^{-1}\smallint_{\mathbb{R}}Q_\lambda
(t_j, z)R_\lambda(t_k, z) \mathrm{d}z$ and $\int_{0}^{t_k}
C_t^0(t_j,z) \mbox{d}z = \alpha^{-1}\smallint_{\mathbb{R}}R_\lambda
(t_j, z)Q_\lambda(t_k, z) \mathrm{d}z$, where each is the adjoint of
the other.

The initial joint GP prior probability distribution for the solution at
a vector of evaluation times ${\bf t_k}$ and its time derivative at a
possibly different vector of evaluation times ${\bf t_j}$, given fixed
hyperparameters $\Psi= (m^0_t,m^0,\alpha,\lambda)$ is,
\begin{align}
\left[
\begin{matrix} u_t({\bf t_j})\\ u({\bf t_k})
\end{matrix}
\right] \sim
{\cal{GP}}\left(\left[
\begin{matrix} m^0_t({\bf t_j})\\ m^0({\bf t_k})
\end{matrix}
\right]
, \left[
\begin{matrix} C_t^0({\bf t_j},{\bf t_j}) & \int_{0}^{{\bf t_k}}
C_t^0({\bf t_j},{\bf z}) \mbox{d}{\bf z} \\ \int_{0}^{{\bf t_k}}
C_t^0({\bf z},{\bf t_j}) \mbox{d}{\bf z} & C^0({\bf t_k},{\bf t_k})
\end{matrix}
\right]\right).
\label{eqn:jointprior}
\end{align}
This prior modelling strategy is straightforwardly generalized to
explicit initial value problems of the form (\ref{eqn:vecfield}) by
defining the prior jointly on the state and the required higher order
derivatives.

As a simple example, consider the second order initial value ODE problem,
\begin{align}
\left\{
\begin{array}{l l l}
u_{tt}(t) = \sin(2t) - u, \quad t \in[0,10], \\
u_t(0) =0, \; {u}(0) = -1.
\end{array}
\right.
\label{eqn:toyIVPsystem}
\end{align}
Its exact solution, $u^*(t) = \{ -4\cos(t) + 2\sin(t) - \sin
(2t) + \cos(t) \}/(4-1)$, is assumed unknown
\textit{a priori}. The first column of Figure \ref{fig:Algo1_Illustration}
shows five draws from the marginal densities of the prior $[
u_t({\bf t}), u({\bf t})]$ 
 defined over a fine grid ${\bf t}\in[0,10]^T$ using a squared exponential covariance with $(\lambda, \alpha) = (0.8, 5)$. Section
\ref{sec:probalgorithm} describes a sequential updating procedure of
the prior model to obtain the probabilistic solution for this example.

\section{Updating prior beliefs for fixed but unknown ODE solution}\label{sec:probalgorithm}

The prior model (\ref{eqn:jointprior}) on the unknown solution and its
derivative will be updated conditionally on model information collected
over the ordered partition ${\bf s}=(s_1,\ldots,s_N)$ of the domain
$[0, L]$. Define a \textit{model interrogation} at time $s_n$ as the
output $\mbox{f}_n \equiv f\{s_n, u^{n-1}(s_n), \theta\}
$ from the right hand side of equation (\ref{eqn:odeivp}) evaluated at
a sample $u^{n-1}(s_n)$ from the marginal prior $[u(s_n) \mid\mbox
{f}_{n-1}]$. Model interrogations are generated sequentially based on
state samples drawn from progressively refined prior models in order to
accommodate fast changing dynamics. Because the ODE model (\ref
{eqn:odeivp}) defines the solution implicitly, we are unable to measure
the exact state or derivative, and therefore both the model
interrogations and their error structure will take into account
increasing uncertainty with time.

This section begins by expanding on these function estimation steps
with full details of the Bayesian updating algorithm. We illustrate the
forward solution algorithm on the simple example (\ref
{eqn:toyIVPsystem}) with known closed form solutions before considering
more realistic systems. Section \ref{sec:complexity} examines
modelling choices for which this procedure can attain computational
complexity of $O(N)$, while Section \ref{sec:rate} establishes the
consistency and first order convergence properties of this approach.
Choice of hyperparameters is addressed in Section \ref
{sec:hyperparms}. Specific extensions to the full generality of (\ref
{eqn:vecfield}) are presented in Sections \ref{sec:Inverse} and  \ref{sec:FsimAndInference}.

\subsection{Model updating based on interrogations from sampled states}\label{sub:update}

This section describes a sequential sampling scheme to interrogate the
model after each update. Model interrogations occur at pre-determined
discretization grid locations ${\bf s}=(s_1,\ldots,s_N)$, while the
resulting sequence of posterior solution models can be computed at a
possibly different vector of time points, ${\bf t}$. There are many
ways to generalize this basic one-step-ahead sampling scheme, as
discussed in Section \ref{sec:discussion}.

\begin{figure}[t!]
\includegraphics[scale=1.05]{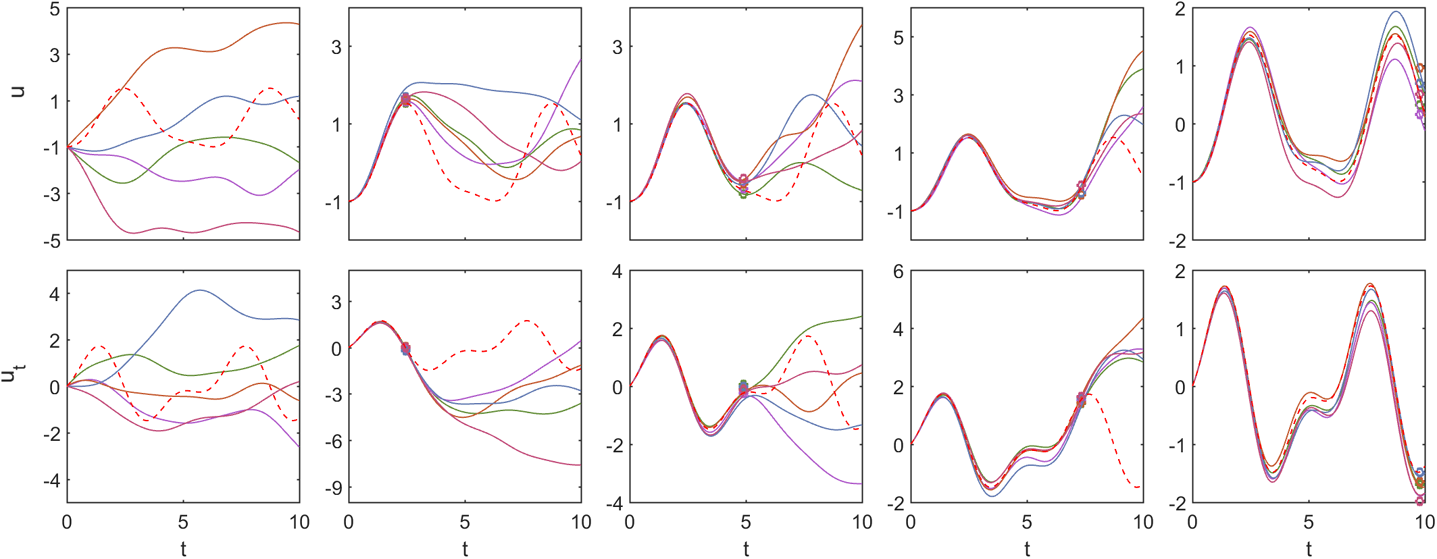}
\caption{Five draws (solid lines) from the marginal densities of (\ref{eqn:posterior_solution})
over the state (first row)
and its derivative (second row) for ODE (\ref{eqn:toyIVPsystem})
after $n = 0, 12, 24, 36, 48$
iterations (from left to right) of Algorithm \ref{alg:ivpsolution}
with discretization grid of size $N=50$; step-ahead samples $u^n$ and $\mbox{f}_n$ are shown as circles;
exact solution is shown as a dotted red line.}
\label{fig:Algo1_Illustration}
\end{figure}

\subsubsection{First update}

The derivative of the exact solution on the boundary, $s_1=0$, can be
obtained exactly by evaluating the right hand side of (\ref
{eqn:odeivp}) at the known initial condition $u^*(s_1)$ as,
\[
\mbox{f}_1 \equiv f\left\{ s_1, u^*(s_1), \theta\right\} = u^*_t(s_1).
\]
The first update of the prior probability model (\ref{eqn:jointprior})
is performed by conditioning on the exact derivative $\mbox{f}_1$.
This yields the joint conditional predictive probability distribution
at evaluation times ${\bf t_j}$ and ${\bf t_k}$,
\begin{align}
\left[
\begin{matrix} u_t({\bf t_j})\\ u({\bf t_k})
\end{matrix}
\; \bigg| \;
\begin{matrix} \mbox{f}_1
\end{matrix}
\right]
\sim
{\cal{GP}}\left(\left[
\begin{matrix} m^1_t({\bf t_j})\\ m^1({\bf t_k})
\end{matrix}
\right]
, \left[
\begin{matrix}
C^1_t({\bf t_j},{\bf t_j})
& \int_{0}^{{\bf t_k}} C^1_t({\bf t_j},{\bf z})\mbox{d}{\bf z}\\
\int_{0}^{{\bf t_k}} C^1_t({\bf z},{\bf t_j}) \mbox{d}{\bf z}
& C^1({\bf t_k}, {\bf t_k})
\end{matrix}
\right]\right),
\label{eqn:jointprior_step1}
\end{align}
\noindent where marginal means and covariances at the vectors of time
${\bf t_j}$ and ${\bf t_k}$ are given by:
\begin{align*}
m^1_t(\mathbf{t_j}) &= m_t^0(\mathbf{t_j}) + C_t^0(\mathbf{t_j},
s_1)C_t^0(s_1,s_1)^{-1} \left\{ \mbox{f}_1 - m^0_t(s_1) \right\},\\
m^1(\mathbf{t_k})
&= m^{0}(\mathbf{t_k}) + C_t^0(s_1,s_1)^{-1}
\smallint_{0}^{\mathbf{t_k}} C_t^{o}({\bf z},s_1) \mbox{d}{\bf z} \,
\left\{\mbox{f}_1 - m_t^{0}(s_1) \right\},\\
C^1_t(\mathbf{t_j},\mathbf{t_j}) &= C_t^0(\mathbf{t_j}, \mathbf
{t_j}) - C_t^0(\mathbf{t_j}, s_1) C_t^0(s_1, s_1)^{-1}
C_t^0(s_1,\mathbf{t_j}),\\
C^1(\mathbf{t_k},\mathbf{t_k}) & = C^{0}(\mathbf{t_k},\mathbf{t_k})
- \left\{ \smallint_{0}^{\mathbf{t_k}} C_t^{0}({\bf z},s_1)\mbox
{d}{\bf z} \,\right\} C_t^0(s_1,s_1)^{-1} \, \left\{\smallint
_{0}^{\mathbf{t_k}} C_t^{0}({\bf z},s_1)\mbox{d}{\bf z} \right\}
^\top.
\end{align*}
Closed forms for means and covariances are obtained by integration
(along with cross-covariances, Supplement D.1). 
 We
refer to this predictive probability measure as the \textit{first
update} of the algorithm and use it as the prior for the next iteration
of the algorithm.

\subsubsection{Second update}

The derivative of the exact solution is unknown beyond $s_1=0$. We
therefore introduce the random variable $\mbox{f}_2$, representing a
model interrogation at the subsequent time step, $s_2>0$. A realization
of $\mbox{f}_2$ is obtained by drawing a sample $u^1(s_2)$ at time
$s_2$ from the marginal predictive distribution of (\ref
{eqn:jointprior_step1}),
\begin{equation}
u(s_2)\mid\mbox{f}_1 \sim\mathcal{N} \left(
m^1(s_2),C^1(s_2,s_2)\right),\nonumber
\end{equation}
and then applying the transformation from the right hand side of equation (\ref{eqn:odeivp}),
\[
\mbox{f}_2 \equiv f\left\{s_2, u^1(s_2), \theta\right\}.
\]
The uncertainty in $u(s_2)\mid\mbox{f}_1$ means that $\mbox{f}_2$ is
not guaranteed to equal to the exact derivative at time $s_2$ (the
function $f$ is a differential operator only for the input $u=u^*$).
Indeed, without the exact solution, the probability distribution of
$\mbox{f}_2$ is unavailable. As with any model for a data generating
process, the likelihood (error model) should be elicited on a
case-by-case basis. Here we provide general guidance on constructing
this likelihood in the absence of additional expert information about
the error model for $\mbox{f}_2$, and all subsequent interrogations
$\mbox{f}_n$. In general, we expect that (i) the discrepancy between
$u_t(t)$ and $f\{t,u(t),\theta\}, t\in[0,L]$ is zero
when $u$ is the exact solution; and (ii) the function $f$ is smooth in
the second argument, therefore discrepancy decreases smoothly with
decreasing uncertainty in the derivative of the sampled state
realization $u^{1}(s_2)$.
These facts suggest a spherically symmetric error model for $\mbox
{f}_2$ with a single mode at $\mbox{f}_2 = u_t(s_2)$ and dispersed as
an increasing function of distance from $s_1=0$. For example, we may
consider the Gaussian error model,
\begin{align}
\mbox{f}_2\mid u_t, \mbox{f}_1 \sim\mathcal{N}\left(u_t(s_2),
r_1(s_2) \right),
\label{eqn:error_model_step_2}
\end{align}
with mean $u_t(s_2)$ and variance $r_1(s_2) = C_t^1(s_2, s_2)$. A
degenerate special case of this error model is discussed at the end of
this section. We may now update (\ref{eqn:jointprior_step1}) by
conditioning on the model interrogation $\mbox{f}_2$ under the error
model (\ref{eqn:error_model_step_2}) to obtain the predictive
probability distribution at the vectors of evaluation times $\mathbf
{t_j}$ and $\mathbf{t_k}$,
\begin{align}
\left[
\begin{matrix} u_t(\mathbf{t_j})\\ u(\mathbf{t_k})
\end{matrix}
\; \bigg| \;
\begin{matrix}
\mbox{f}_2,\mbox{f}_1
\end{matrix}
\right] \sim
{\cal{GP}}\left(\left[
\begin{matrix} m^2_t(\mathbf{t_j})\\ m^2(\mathbf{t_k})
\end{matrix}
\right]
, \left[
\begin{matrix}
C^2_t(\mathbf{t_j}, \mathbf{t_k})
& \int_{0}^{\mathbf{t_k}} C^2_t(\mathbf{t_j},{\bf z})\mbox{d}{\bf
z}\\
\int_{0}^{\mathbf{t_j}} C^2_t({\bf z},\mathbf{t_k}) \mbox{d}{\bf z}
& C^2(\mathbf{t_k}, \mathbf{t_k})
\end{matrix}
\right]\right).
\label{eqn:jointprior_step3}
\end{align}
Defining $g_2 = C_t^{1}(s_2,s_2) + r_1(s_2)$, the marginal means and
covariances at the vectors of evaluation times ${\bf t_j}$ and ${\bf
t_k}$ are given by,
\begin{align*}
m_t^2(\mathbf{t_j}) &= m_t^{1}(\mathbf{t_j}) + C_t^{1}(\mathbf
{t_j},s_2)\,g_2^{-1} \left\{\mbox{f}_2 - m_t^{1}(s_2) \right\},\\
m^2(\mathbf{t_k}) &= m^{1}(\mathbf{t_k}) +
\smallint_{0}^{\mathbf{t_k}} C_t^{1}({\bf z},s_n) \mbox{d}{\bf z} \,
g_2^{-1} \, \left\{\mbox{f}_2 - m_t^{1}(s_2) \right\},\\
C_t^2(\mathbf{t_j},\mathbf{t_j}) & = C_t^{1}(\mathbf{t_j},\mathbf
{t_j}) - C_t^{1}(\mathbf{t_j},s_2)\,g_2^{-1} C_t^{1}(s_2,\mathbf
{t_j}), \\
C^2(\mathbf{t_k},\mathbf{t_k}) & = C^{1}(\mathbf{t_k},\mathbf{t_k})
- \left\{ \smallint_{0}^{\mathbf{t_k}} C_t^{1}({\bf z},s_2)\mbox
{d}{\bf z} \,\right\} \,g_2^{-1}\, \left\{\smallint_{0}^{\mathbf
{t_k}} C_t^{1}({\bf z},s_2)\mbox{d}{\bf z} \right\}^\top,
\end{align*}
where state mean, covariance, and derivative cross-covariances can be
obtained analytically via integration (see Supplement D.1).

\subsubsection{Subsequent updates}\label{seq:subsequent}

Subsequent updates, $2<n\leq N$, begin by drawing a sample
$u^{n-1}(s_{n})$ at time $s_{n}$ (illustrated by circles in the top row of Figure \ref
{fig:Algo1_Illustration}) from the marginal predictive posterior
$[u(s_n) \mid\mbox{f}_{n-1},\ldots,\mbox{f}_1 ]$, {whose mean
and covariance are given in Algorithm \ref{alg:ivpsolution}}. A
realization of the model interrogation $\mbox{f}_n$ is constructed by
applying the model transformation (\ref{eqn:odeivp}) to
$u^{n-1}(s_{n})$. The update is performed relative to the prior
$[u_t, u \mid\mbox{f}_{n-1},\ldots,\mbox{f}_1 ]$, by
assimilating $\mbox{f}_n$ under the Gaussian error model,
\begin{align}
\mbox{f}_n \mid u_t(s_n), \mbox{f}_{n-1},\ldots, \mbox{f}_1 \sim
\mathcal{N} \left\{u_t(s_n), r_{n-1}(s_n) \right\}, \quad1\leq n
\leq N,
\label{eqn:error_model_step_n}
\end{align}
where, e.g., $r_{n-1}(s_n)=C_t^{n-1}(s_n,s_n)$, to obtain the
predictive probability distribution at the vectors of evaluation times
$\mathbf{t_j}$ and $\mathbf{t_k}$,
\begin{align}
\left[
\begin{matrix} u_t(\mathbf{t_j})\\ u(\mathbf{t_k})
\end{matrix}
\; \bigg| \;
\begin{matrix}
\mbox{f}_n,\ldots,\mbox{f}_1
\end{matrix}
\right] \sim
{\cal{GP}}\left(\left[
\begin{matrix} m^n_t(\mathbf{t_j})\\ m^n(\mathbf{t_k})
\end{matrix}
\right]
, \left[
\begin{matrix}
C^n_t(\mathbf{t_j}, \mathbf{t_j})
& \int_{0}^{\mathbf{t_k}} C^n_t(\mathbf{t_j},{\bf z})\mbox{d}{\bf
z}\\
\int_{0}^{\mathbf{t_j}} C^n_t({\bf z},\mathbf{t_k}) \mbox{d}{\bf z}
& C^n(\mathbf{t_k}, \mathbf{t_k})
\end{matrix}
\right]\right).
\label{eqn:jointprior_step3n}
\end{align}
Defining $g_{n} = C_t^{n-1}(s_n,s_n) + r_{n-1}(s_n)$, the resulting
joint predictive probability distribution is Gaussian with marginal
means and covariances:
\begin{align}
m_t^n({\bf t_j})
&= m_t^{n-1}({\bf t_j}) +
C_t^{n-1}({\bf t_j},s_n)\,g_n^{-1}\left\{ \mbox{f}_n - m_t^{n-1}(s_n)
\right\} ,\nonumber\\
m^n({\bf t_k})
&= m^{n-1}({\bf t_k}) + g_n^{-1}
\smallint_{0}^{{\bf t_k}} C_t^{n-1}({\bf z},s_n) \mbox{d}{\bf z} \,
\left\{\mbox{f}_n - m_t^{n-1}(s_n) \right\},\nonumber\\
C_t^n({\bf t_j},{\bf t_k}) & = C_t^{n-1}({\bf t_j},{\bf t_k}) -
C_t^{n-1}({\bf t_j},s_n)\, g_n^{-1}C_t^{n-1}(s_n,{\bf t_k}),\nonumber\\
C^n({\bf t_j},{\bf t_k}) & = C^{n-1}({\bf t_j},{\bf t_k}) - g_n^{-1}
\smallint_{0}^{{\bf t_j}} C_t^{n-1}({\bf z},s_n)\mbox{d}{\bf z} \,
\left\{\smallint_{0}^{{\bf t_k}} C_t^{n-1}({\bf z},s_n)\mbox{d}{\bf
z} \right\}^\top,\nonumber
\end{align}
and with state mean, covariances, and derivative cross-covariances
again obtained via integration (see Algorithm \ref{alg:ivpsolution}
and Supplement D.1).  
Here the state and derivative,
$u$ and $u_t$ respectively, may be of arbitrary dimension. Under a
bounded support covariance, $C^n({\bf t_j},{\bf t_k})$ and $C_t^n({\bf
t_j},{\bf t_k})$ become sparse, band diagonal matrices.

After $N$ updates have been performed, we obtain a joint posterior
probability distribution $[u_t, u \mid{\bf f} ]$ for the
unknown solution and its derivative, conditional on realizations of one
trajectory of the model interrogation vector $\mbox{\bf f}$. Draws
from the marginal model (\ref{eqn:posterior_solution}) can be obtained
via Monte Carlo, as described in Algorithm \ref{alg:ivpsolution},
yielding realized trajectories from the forward model (\ref
{eqn:posterior_solution}) at desired time locations ${\bf t}$. Five
draws from the joint distribution of (\ref{eqn:posterior_solution})
and the interrogations $\mbox{f}_n$ at iterations $n = 12, 24, 36, 48$
are illustrated in the second through fifth columns of Figure \ref
{fig:Algo1_Illustration} for the second-order ODE initial value problem
(\ref{eqn:toyIVPsystem}). For notational simplicity, we will hereafter
assume that the temporal locations of interest, ${\bf t}$, form a
subset of the discretization grid locations, ${\bf s}$.

\subsubsection{Point mass likelihood on model interrogations}

A special case of this procedure models each interrogation, $\mbox
{f}_n$, as the exact unknown derivative, $u^*_t(s_j)$, regardless of
grid size, $N$. Thus, model interrogations are interpolated in the
derivative space by replacing (\ref{eqn:error_model_step_2}) and (\ref
{eqn:error_model_step_n}) with the point mass density:
\[
P\left(u_t(s_{n}) =\mbox{f}_n \right)=1, \quad1\leq n\leq N.
\]
Because sampled process realizations can still be inconsistent with the
solution, computational implementation of this alternative scheme often
requires using an arbitrary covariance nugget increasing with time, or
the use of a rougher covariance function than the smoothness implied by
the differential equation model. For these reasons we will restrict our
attention to explicitly modelling the interrogation error, e.g. (\ref
{eqn:error_model_step_2}) and (\ref{eqn:error_model_step_n}).

\begin{algorithm}[t]
\caption{For the initial value problem (\ref{eqn:odeivp}), draw one
sample from the forward model (\ref{eqn:posterior_solution}) over
${\bf t} = (t_1,\cdots,t_T)$, given $\Psi$ and discretization grid
${\bf s} = (s_1, \cdots, s_N)$.}
\label{alg:ivpsolution}
\begin{algorithmic}
\STATE At time $s_1 = 0$ initialize the derivative $\mbox{f}_{1}= f \{ s_1, u^*(0), \theta\}$, and define $m^0, m^0_t, C^0, C^0_t$ as above;
\FOR{$n = 1 : N$}
\STATE If $n=1$, set $g_1=C_t^{0}(s_1,s_1)$, otherwise set
$g_n=C_t^{n-1}(s_n,s_n)+r_{n-1}(s_n)$;
\STATE Compute for each modelled system component,
\begin{align*}
m^n({\bf s})
&= m^{n-1}({\bf s}) + g_n^{-1}
\smallint_{0}^{{\bf s}} C_t^{n-1}({\bf z},s_n) \mbox{d}{\bf z} \,
\left\{\mbox{f}_n - m_t^{n-1}(s_n) \right\},\\
m_t^n({\bf s})
&= m_t^{n-1}({\bf s}) + g_n^{-1} \;
C_t^{n-1}({\bf s},s_n) \left\{\mbox{f}_n- m_t^{n-1}(s_n) \right\},\\
C^n({\bf s},{\bf s}) & = C^{n-1}({\bf s},{\bf s}) - g_n^{-1} \smallint
_{0}^{{\bf s}} C_t^{n-1}({\bf z},s_n)\mbox{d}{\bf z} \, \left\{
\smallint_{0}^{{\bf s}} C_t^{n-1}({\bf z},s_n)\mbox{d}{\bf z} \right
\}^\top,\\
C_t^n({\bf s},{\bf s}) & = C_t^{n-1}({\bf s},{\bf s}) - g_n^{-1}
C_t^{n-1}({\bf s},s_n)\, C_t^{n-1}(s_n,{\bf s}),\\
\smallint_{0}^{{\bf s}} C_t^n({\bf z},{\bf s}) \mbox{d}{\bf z} & =
\smallint_{0}^{{\bf s}}C_t^{n-1}({\bf z},{\bf s}) \mbox{d}{\bf z} -
g_n^{-1} \smallint_{0}^{{\bf s}} C_t^{n-1}({\bf z},s_n) \mbox{d}{\bf
z} \, C_t^{n-1}(s_n,{\bf s});
\end{align*}
\IF{$n < N$}
\STATE Sample one-step-ahead realization $u^{n}(s_{n+1})$ from the
predictive distribution on the state for each modelled system component,
\begin{align*}
u^{n}(s_{n+1}) \,{\sim}\, p \left\{ u^n(s_{n+1}) \mid\mbox{f}_n,\ldots
,\mbox{f}_1 \right\}
{=}\, \mathcal{N} \left\{ u^n(s_{n+1}) \mid m^n(s_{n+1}),
C^n(s_{n+1},s_{n+1}) \right\}\!,
\end{align*}
and interrogate the model by computing $\mbox{f}_{n+1} ={f}\{
s_{n+1}, u^{n}(s_{n+1}), \theta\}$;
\ENDIF
\ENDFOR
\STATE Sample and return
$\mathbf{u} = (u^N(t_1),\cdots,u^N(t_T))\sim\mathcal
{GP}\{ m^N({\bf t}), C^N({\bf t},{\bf t})\}$, where ${\bf
t}\subset{\bf s}$.
\end{algorithmic}
\end{algorithm}

\subsection{Computational complexity} \label{sec:complexity}

A reasonable modelling assumption for the state is that temporal
correlation decays to zero with distance. Indeed, most deterministic
numerical methods make the analogous assumption by defining a lag of
$M\geq1$ time steps, beyond which the solution has no direct effect.
This Markovian assumption both adds flexibility for modelling very
fast-changing dynamics and drastically reduces computational cost.
Under a compactly supported covariance structure, Algorithm \ref
{alg:ivpsolution} requires only $O(N)$ operations by allowing
truncation of the weight matrices employed in the mean and covariance
updates. In contrast, when employing a covariance structure with
unbounded support, Algorithm \ref{alg:ivpsolution} attains
computational complexity proportional to the cube of the number of time
steps, since the number of operations can be written as a finite sum
over $1\leq n\leq N$ with each iteration requiring an order $O(n^2)$ operations.

A computational comparison of Algorithm \ref{alg:ivpsolution} with a
similarly sampled explicit first order solver is provided in the Supplement C.  
Although probability models for
the solution attain the same computational scaling, the probabilistic
method is more expensive. However, the model of solution uncertainty
makes up for its added computational cost in practical applications
where discretization grids cannot be refined enough to ensure
negligible discretization effects. Within the inverse problem, the
adaptive estimation of length-scale can substantially increase solver
accuracy (see Figure 10 
in Supplement C). 

\subsection{Convergence} \label{sec:rate}

We now consider the rate at which the stochastic process with
probability distribution (\ref{eqn:posterior_solution}) generated via
Algorithm \ref{alg:ivpsolution} concentrates around the exact solution
of (\ref{eqn:odeivp}). The rate of convergence is as fast as the
reduction of local truncation error for the analogously sampled
one-step explicit Euler method. Proof is provided in the 
Supplement D.2.

\begin{thm}
\label{thm:meanconvergence}
Consider a function $f(t,u(t)): [0,L] \times\mathbb{R} \to\mathbb
{R}$ that is continuous in the first argument and Lipschitz continuous
in the second argument. The stochastic process with probability
distribution (\ref{eqn:posterior_solution}) obtained using Algorithm
\ref{alg:ivpsolution} with $r_n(s) = 0, s\in[0,L]$ and a covariance
kernel satisfying conditions (D.6) and (D.7) 
 in Supplement D.3,  
converges in $L^1$ to the
unique solution satisfying (\ref{eqn:odeivp}) as the maximum
discretization grid step size, $h$, the length-scale, $\lambda$, and
the prior variance, $\alpha^{-1}$, tend to zero. The rate of
convergence is $O(h)$, proportional to the decrease in the maximum step
size, $h$.
\end{thm}

\subsection{The choice of hyperparameters} \label{sec:hyperparms}

Analogously with choosing a numerical solver, the choice of
probabilistic solver hyperparameters must be based on prior knowledge
about the solution, such as its smoothness. Moreover, as with many
nonparametric function estimation problems, asymptotic results on
convergence may be used to guide our choice of covariance
hyperparameters within the forward problem. For example, for the
explicit first order sampling scheme, the result of Theorem \ref
{thm:meanconvergence} suggests setting the prior variance, $1/\alpha$,
and the length-scale, $\lambda$, proportionally to the maximum
distance, $h$, between subsequent discretization grid locations. When
experimental data is available, it also becomes possible to obtain
marginal posterior densities for the hyperparameters as part of the
calibration problem (see Section \ref{sec:Inverse}), allowing data
driven adaptation over the space of model parameters. In contrast,
deterministic numerical differential equation solvers effectively fix
the analogues of these hyperparameters, such as the quadrature rule or
the step number, both within the forward and inverse problems.

\section{Exact Bayesian posterior model inference}\label{sec:Inverse}

The discretization uncertainty model in Section \ref
{sec:probalgorithm} may now be incorporated into the calibration
problem. We wish to infer model parameters $\theta$ defining states
$\mathbf{u} = (u(t_1),\ldots,\break u(t_T))^\top$ from
measurement data $\mathbf{y} = (y(t_1),\ldots,y(t_{T}))^\top$,
via the hierarchical model (\ref{eqn:posterior_probabilistic}),
\begin{align*}
\begin{array}{rrl}
[ \mathbf{y} \mid\mathbf{u}, \theta, \Sigma] & \propto& \rho
\left\{ \mathbf{y} - H \left(\mathbf{u}, \theta\right)\right\},\\
\left[ \mathbf{u} \mid\theta, \Psi, N \right] & = & p\left
(\mathbf{u} \mid\theta, \Psi, N \right),\\
\left[ \theta, \Psi, \Sigma \right] & = & \pi(\theta,\Psi,
\Sigma),
\end{array}
\end{align*}
where the transformation $H$ maps the differential equation solution to
the observation space, $\rho$ is a probability density, $\pi$ is a
prior density, $\Sigma$ defines the observation process, such as the
error structure, and $\theta$ may include unknown initial conditions
or other model components. We describe a Markov chain Monte Carlo
procedure targeting the joint posterior of the state and unknown
parameters conditional on a vector of noisy observations from $[
\mathbf{y} \mid\mathbf{u}, \theta, \Sigma]$.
Section \ref{subsec:jakstatexample} describes a case where states are
indirectly observed through a nonlinear transformation $H$.

\begin{algorithm}[ht!]
\caption{Draw $K$ samples from posterior $p ( \theta, \mathbf
{u}\mid\mathbf{y}, \alpha, \lambda, \Sigma) $ given
observations of the transformed solution of a differential equation
with unknown parameters $\theta$.}
\label{alg:fullyprobinf}
\begin{algorithmic}
\STATE
\STATE Initialize $\theta, \alpha, \lambda\sim\pi(\cdot)$ where
$\pi$ is the prior density;
\STATE Using a probabilistic solver (e.g., Algorithm \ref{alg:ivpsolution})
simulate a vector of realizations, $\mathbf{u}$ conditional on $\theta
, \alpha,\lambda$, over the size N discretization grid;
\FOR{$k = 1:K$}
\STATE Propose $\theta', \alpha', \lambda'\sim q(\cdot\mid \theta
, \alpha, \lambda)$ where $q$ is a proposal density;
\STATE Using a probabilistic solver (e.g., Algorithm \ref{alg:ivpsolution})
simulate a vector of realizations, $\mathbf{u}'$ conditional on
$\theta', \alpha', \lambda'$, over the size N discretization grid;
\STATE Compute the rejection ratio,
\begin{align*}
\rho & =
\frac{p \left( \, \mathbf{y} \mid\mathbf{u}', \theta', \Sigma\,
\right)}{ p \left( \, \mathbf{y} \mid\mathbf{u}, \theta, \Sigma\,
\right)}\; \frac{\pi(\, \theta'\, , \alpha' , \lambda'\,)}{\pi(
\,\theta, \alpha, \lambda\,)}\;\frac{q(\, \theta\, , \alpha,
\lambda\mid \theta' , \alpha', \lambda' \,)}{q(\, \theta'\, ,
\alpha' , \lambda' \mid \theta, \alpha, \lambda\,)} ,
\end{align*}
where the ratio of marginal densities, $p(\mathbf{u} \mid{\bf
f},\theta,\alpha,\lambda)$ to $p(\mathbf{u}'\mid{\bf f}', \theta
',\alpha',\lambda')$, is omitted because the realizations $\mathbf
{u}$ and $\mathbf{u}'$ are simulated directly from these densities;

\STATE{Sample $v\sim\mbox{U}[0,1]$}
\IF{ $v<\min( \rho,1 ) $}
\STATE Update $( \theta, \alpha, \lambda) \leftarrow
( \theta'\, , \alpha' , \lambda')$;
\STATE Update $\mathbf{u}\leftarrow\mathbf{u}'$;
\ENDIF
\STATE Return $(\theta, \alpha, \lambda, \mathbf{u})$.
\ENDFOR
\end{algorithmic}
\end{algorithm}

Algorithm \ref{alg:fullyprobinf} targets the posterior distribution
(\ref{eqn:posterior_probabilistic}) by generating forward model
proposals via conditional simulation from Algorithm \ref
{alg:ivpsolution} (Supplement B  
provides a parallel
tempering implementation). Within Algorithm \ref{alg:fullyprobinf},
proposed sample paths are conditionally simulated from $p(\mathbf
{u}\mid\theta, \Psi) = \int p(\mathbf{u}, \mathbf{f}
\mid\theta, \Psi) \mbox{d}{\bf f}$, removing the dependence
of the acceptance probability on this intractable marginal density.
Such partially likelihood-free Markov chain Monte Carlo methods \citep{MarjoramEtAl2003}
are widely used for inference on stochastic differential equations
\citep[e.g.,][]{GolightlyWilkinson2011}.

As with numerical based approaches, computation of the forward model is
the rate limiting step in Algorithm \ref{alg:fullyprobinf}.
Computational savings can be obtained by exploiting the structure of
the probabilistic solver, such as by targeting progressively refined
posterior models in a Sequential Monte Carlo framework, or using
ensemble Markov chain Monte Carlo methods \citep{Neal2011} when mixture components of
(\ref{eqn:posterior_solution}) can be quickly computed. We now have
the components for accounting for discretization uncertainty in
differential equation models of natural and physical systems. The
remainder presents a calibration problem illustrating posterior
sensitivity to discretization uncertainty.

\subsection{Inference on a delay initial function model} \label
{subsec:jakstatexample}

The JAK-STAT mechanism describes a series of reversible biochemical
reactions of STAT-5 transcription factors initiated by binding of the
Erythropoietin (Epo) hormone to cell surface receptors \citep
{PellegriniDusanterFourt1997}. After gene activation occurs within the
nucleus, the transcription factors revert to their initial state,
returning to the cytoplasm to be used in the next activation cycle.
This last stage is not well understood and is replaced by the unknown
time delay, $\tau>0$. One model describes changes in 4 reaction states
of STAT-5 via the nonlinear delay initial function problem,
\begin{align}
\left\{
\begin{array}{l l l l}
u_t^{(1)}(t) &= - &k_1\, u^{(1)}(t) \;\mbox{Epo} R_A(t) + \; 2\,k_4\,
u^{(4)}(t-\tau), \quad\quad& t\in[0,60], \\
u_t^{(2)}(t) &= &k_1\, u^{(1)}(t) \;\mbox{Epo} R_A(t) -k_2\,
{u^{(2)}}^2(t), & t\in[0,60],\\
u_t^{(3)}(t) &= -&k_3\, u^{(3)}(t) +\; \tfrac{1}{2}\, k_2 \,
{u^{(2)}}^2(t), & t\in[0,60],\\
u_t^{(4)}(t) &= &k_3\, u^{(3)}(t ) -k_4\, u^{(4)}(t-\tau), & t\in
[0,60],\\
u(t) &=& \phi(t), &t \in[-\tau,0].
\end{array}
\right.
\label{eqn:JackStatDDE}
\end{align}
\noindent The initial function components
$\phi^{(2)}(t) = \phi^{(3)}(t) =\phi^{(4)}(t)$ are identically zero,
and the constant initial function $\phi^{(1)}(t)$ is unknown. The
states for this system cannot be measured directly, but are observed
through a nonlinear transformation,
$H: \mathbb{R}^3 \times\Theta\to\mathbb{R}^4$, defined as,
\begin{align}
H\left\{u(t), \theta\right\} =
\left(
\begin{array}{l}
k_5 \left\{ u^{(2)}(t) + 2u^{(3)}(t) \right\}\\ [4.5pt]
k_6 \left\{ u^{(1)}(t) + u^{(2)}(t) + 2u^{(3)}(t) \right\}\\[4.5pt]
u^{(1)}(t) \\
u^{(3)}(t) \left\{u^{(2)}(t) + u^{(3)}(t)\right\}^{-1}
\end{array}
\right),
\label{eqn:JackStatObs}
\end{align}
\noindent and parameterized by the unknown scaling factors $k_5$ and
$k_6$. The indirect measurements of the states,
\[
y^{(j)}(t_{j,i}) = H^{(j)}\left\{ u(t_{j,i}), \theta\right\} +
\varepsilon^{(j)}(t_{j,i}), \quad1\leq i \leq T_j, \; 1\leq j \leq4,
\]
are assumed to be contaminated with additive zero mean Gaussian
noise,\break
$\{\varepsilon^{(j)}(t_{j,i})\}_{1\leq i \leq T_j, \; 1\leq j \leq
4}$, with experimentally determined standard deviations. Analysis is
based on experimental data \citep{SwameyeEtAl2003,RaueEtAl2009}. As
per \citet{RaueEtAl2009}, the forcing function $\mbox{Epo} R_A$
shares the same smoothness as the solution state (piecewise linear
first derivative) as modelled by a GP interpolation.

This data has been used for calibrating the JAK-STAT pathway mechanism
by a number of authors \citep
[e.g.,][]{SchmidlEtAl2013,SwameyeEtAl2003,
RaueEtAl2009,CampbellChkrebtii2013} under a variety of modelling
assumptions and inferential approaches. The inaccuracy and
computational constraints of numerical techniques required to solve the
system equations have led some authors to resort to coarse ODE
approximations, or to employ semiparametric Generalized Smoothing
methods. This motivates analysis of the structure and propagation of
discretization error through the inverse problem. A further issue is
that the model (\ref{eqn:JackStatDDE}) and its variants may be
misspecified, which must be first uncoupled from discretization effects
for further study. Indeed, the above cited analyses conclude that the
model is not flexible enough to fit the available data, however it is
not clear how much of this misfit is due to model discrepancy, and how
much may be attributed to discretization error.

\begin{figure}[ht!]
\includegraphics[scale=0.97]{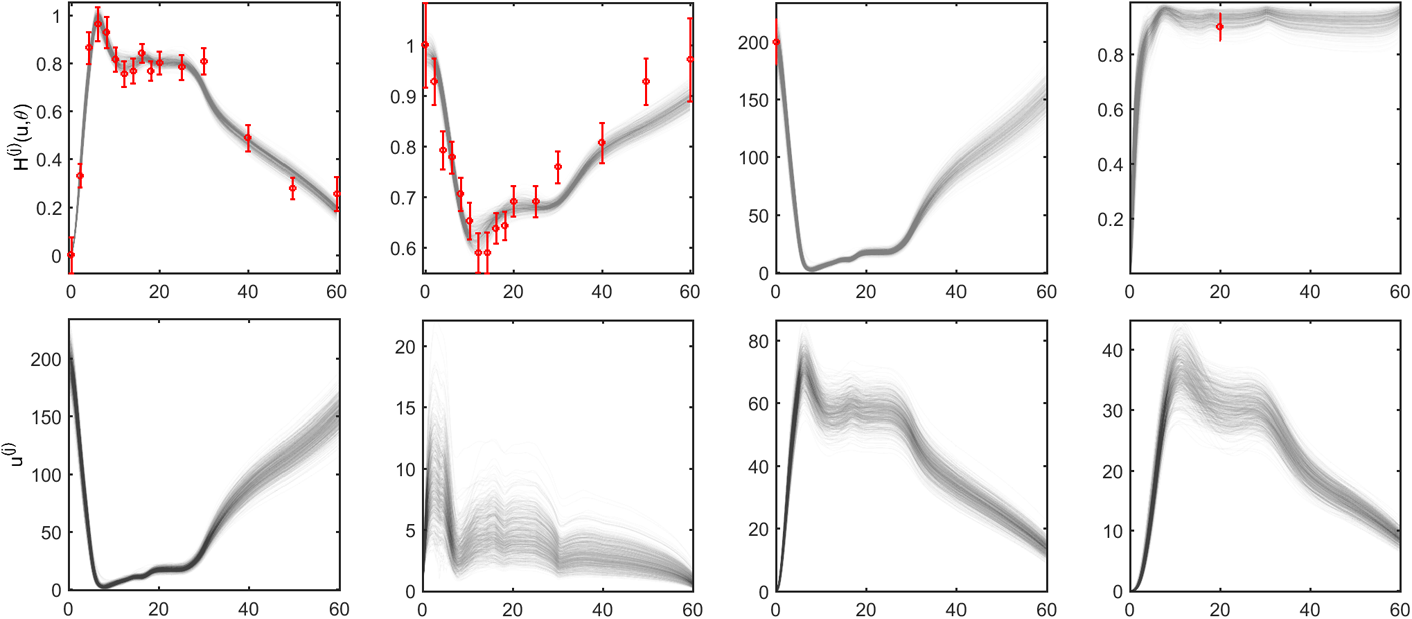}
\caption{Sample paths (grey lines) of the observation processes (first
row) obtained by
transforming a sample from the marginal posterior distribution of the
states (second row)
by the observation function (\ref{eqn:JackStatObs}); experimental data
(red points) and
error bars showing two standard deviations of the measurement error.}
\label{fig:JakStatSolver}
\end{figure}

\begin{table}
\centering
\begin{tabular}{ l l }
\hline
Model component & Prior \\\hline
$\tau$ & $\chi^2_6$ \\
$k_i, \; i = 1,\ldots, 6$ & $\mbox{Exp}(1)$ \\
$\lambda_i, \; i = 1,\ldots, 4$ & $\chi^2_1$\\
$\alpha_i + 100, \; i = 1,\ldots, 4 \quad$ & $\mbox{Log-}\mathcal
{N}(10,1 )$ \\
$u^{(1)}(0) $ & $N ( y^{(3)}(0), 40^2 )$ \\
$u_t^{(i)}, \; i = 1,\ldots, 4$ & $\mathcal{GP}(0, C_t
)$ \\
$u^{(i)}, \; i = 1,\ldots, 4 $ & $1\{i=1\} u^{(1)}(0) + \smallint
_{0}^t u_t^{(i)}(s)\,\mbox{d}s $ \\ \hline
\end{tabular}
\caption{Priors on the states and parameters, assumed to be \textit
{a priori} independent.}
\label{table:priors}
\centering
\end{table}

The probabilistic approach for forward simulation of delay initial
function problems is described in Supplement B,
extending the updating strategy to accommodate current and lagged state
process realization information. The probabilistic solution is
generated using an equally spaced discretization grid of size $N =
500$, and a covariance structure based on the uniform covariance
kernel. Prior distributions are listed in\vadjust{\eject} Table \ref{table:priors}.
Markov chain Monte Carlo is used to obtain samples from the posterior
distribution (\ref{eqn:posterior_probabilistic}).\break A~parallel tempering
sampler \citep{Geyer1991} described in Supplement B 
efficiently traverses the parameter space of this
multimodal posterior distribution.

A sample from the marginal posterior of state trajectories and the
corresponding observation process are shown in Figure \ref
{fig:JakStatSolver}. As expected, the posterior trajectories, which incorporate discretization uncertainty
in the forward problem, appear more diffuse than those obtained in
e.g., \citet{SchmidlEtAl2013}, while more flexibly describing
the observed data, with the exception of a small number of outliers.
Correlation plots and kernel density estimates with priors for the
marginal posteriors of rate parameters $k_1,k_2,k_3,k_4$, observation
scaling parameters $k_5,k_6$, time lag $\tau$, initial function value
$u^{(1)}(0)$, and hyperparameters for the solution uncertainty $\Psi
=(\alpha,\lambda)$ are shown in Figure \ref
{fig:JakStatModCorrs}. All parameters with the exception of the prior
precision $\alpha$ are identified by the data, while the rate
parameter $k_2$ appears to be only weakly identified. We observe
strong correlations between parameters, consistent with previous
studies on this system. For example, there is strong correlation among
the scaling parameters $k_5, k_6$ and the initial first state
$u^{(1)}(0)$. Importantly, there is a relationship between the
length-scale $\lambda$ of the probabilistic solver and the first,
third and fourth reaction rates. Furthermore, this relationship has a
nonlinear structure, where the highest parameter density region seems
to change with length scale implying strong sensitivity to the solver
specifications. The length-scale is the probabilistic analogue, under a
bounded covariance, of the step number in a numerical method. However,
in analyses based on numerical integration, the choice of a numerical
technique effectively fixes this parameter at an arbitrary value that
is chosen a priori. The result here suggests that for this problem, the
inferred parameter values are highly and nonlinearly dependent on the
choice of the numerical method used. Because we have explicitly
quantified discretization uncertainty, the observed lack of fit may suggest
additional model refinement may be required.

\begin{figure}
\includegraphics{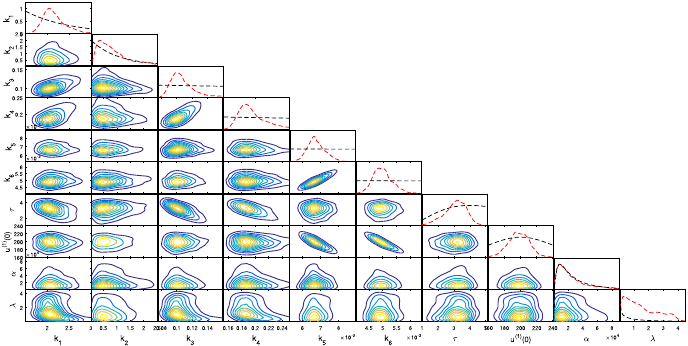}
\caption{50,000 posterior samples of the model parameters using a
probabilistic solver with a
grid of size 500, generated using parallel tempering Algorithm 7 
 with ten chains.
Prior probability densities are shown in black. Note the complex
posterior correlation structure
between model parameters and solver hyperparameters across the lower rows.}
\label{fig:JakStatModCorrs}
\end{figure}

\section{Applications} \label{sec:FsimAndInference}

We extend probability modelling of solution uncertainty to a variety of
forward problems: chaotic dynamical systems, an ill-conditioned mixed
boundary value problem, a~stiff, high-dimensional initial value
problem, and extensions to PDEs.

\subsection{Discretization uncertainty under chaotic dynamics} \label
{sec:lorenz}

The classical ``Lorenz63'' initial value problem \citep{Lorenz1963} is
a three-state ODE model of convective fluid motion induced by a
temperature difference between an upper and lower surface. States
$u^{(1)}, u^{(2)}$ and $u^{(3)}$ are proportional to the intensity of
the fluid's convective motion, the temperature difference between (hot)
rising and (cold) descending currents, and the deviation of the
vertical temperature profile from linearity respectively. The model
describing the time-evolution of their dynamics,
\begin{align}
\left\{
\begin{array}{llll}
u^{(1)}_t &= - \sigma( u^{(1)} + u^{(2)}) , & \quad\quad t\in[0,20],
\\
u^{(2)}_t &= - r u^{(1)} - u^{(2)} - u^{(1)}u^{(3)},& \quad\quad t\in
[0,20],\\
u^{(3)}_t &= u^{(1)}u^{(2)} - b u^{(3)},& \quad\quad t\in[0,20],\\
u & = u^*, & \quad\quad t = 0,\\
\end{array}
\right.
\end{align}
\noindent depends on dimensionless parameters $\sigma,r$ and $b$. The
standard choice of $\theta= (\sigma,r,b)=(10, 8/3, 28)$ with initial
state $u^*(0)= ( -12,-5,38)$ falls within the chaotic regime. This
system has a unique solution whose trajectory lies on a bounded region
of the phase space \citep[e.g.,][pp. 271--272]{Robinson2001}, and is
unknown in closed form.

A sample of 1000 trajectories generated from the probabilistic solution
for this system is shown in Figure \ref{fig:lorenzerrors} given
hyperparameters $N = 5001, \alpha= N, \lambda=2h, h=20/N$, constant
prior mean function and a squared exponential covariance structure.
There is a short time window within which there is negligible
uncertainty in the solution, but the accumulation of discretization
uncertainty results in rapidly diverging, yet highly structured errors.
Figure \ref{fig:lorenzerrors} shows posterior samples in the 3 state
dimensions (top row) and density estimates of the samples for
$u^{(1)}(t)$ (bottom row) at 4 distinct evaluation times. The
structured uncertainty in the forward model (\ref
{eqn:posterior_solution}) is not simply defined by widening of the high
posterior density region, but instead exhibits a behaviour that is
consistent with the dynamics of the underlying stable
attractor.\looseness=1

\subsection{Solution multiplicity in a mixed boundary value problem}
\label{sec:MBVP}

\begin{figure}
\includegraphics[scale=1.1]{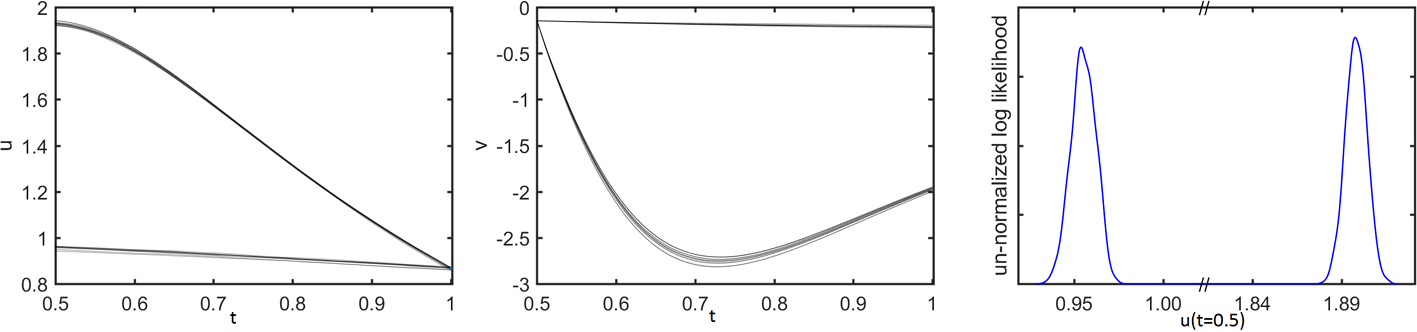}
\caption{$100$ posterior realizations from (\ref{eqn:MBVposterior}).
Left and centre: states,
$u$ and $v$, evaluated over a fine grid, ${\bf t}$, over the domain,
$\mathcal{D} = [.5,1]$.
Right: marginal un-normalized log posterior density over the unknown
initial condition $u(t=.5)$.}
\label{fig:mbvpost}
\end{figure}

Consider a special case the Lane--Emden ODE mixed boundary problem,
which is used to model the density, $u$, of gaseous spherical objects,
such as stars, as a function of radius, $t$, (\citealt{Shampine2003}).
The second order system can be rewritten as a system of a first order
equations on the interval $[.5,1]$ with mixed boundary values,
\begin{align}
\left\{
\begin{array}{ll}
u_t       & =\, v,               \\
v_t       & =\, -2vt^{-1} - u^5, \\
u(t = 1)  & =\, \sqrt{3}/2,      \\
v(t = .5) & =\, -288/2197,
\end{array}
\right.
\label{eqn:Emden1}
\end{align}
where the notation $u(t=1)$ and $v(t=.5)$ is used to emphasize that
boundary conditions are different for each state. The two leftmost
panels of Figure \ref{fig:mbvpost} illustrate multimodality in the
posterior solution, when multiple functions satisfy the mixed boundary
value problem. These high posterior density regions concentrate as the
number of grid points and the prior precision grow, suggesting the
potential for defining intermediate target densities for sequential and
particle Markov chain Monte Carlo schemes to sample realizations from
the posterior distribution over unknown model parameters within the
calibration problem. This example highlights the need for accurately
modelling functional rather than point-wise solution uncertainty.

Mixed boundary value problems introduce challenges for many numerical
solvers, which employ optimization and the theory of ODE initial value
problems to estimate the unknown initial state, here $u(t=.5)$, subject
to the state constraints at either boundary of the domain of
integration, here $[.5,1]$. The optimization over $u(t=.5)$ is
performed until the corresponding initial value problem solution
satisfies the known boundary condition, $u^*(t=1)$, to within a user
specified tolerance. Considering the mixed boundary value problem (\ref
{eqn:Emden1}), the probabilistic solution consists of an inference
problem over the unspecified initial value, $u(t=.5)$, given its known
final value $u^*(t=1)$. The likelihood naturally defines the mismatch
between the boundary value $u(t=1)$, obtained from the realized
probabilistic solution given $u^*(t=1)$, as follows,
\begin{align}
u^*(t=1) \mid u(t=.5), v^*(t=.5)
& \; \sim\; \mathcal{N} \left\{ m^{N(u)}(t=1), C^{N(u)}(t=1,t=1)
\right\},
\label{eqn:endptlikelihood}
\end{align}
\noindent where $m^{N(u)}(t=1)$ and $C^{N(u)}(t=1,t=1)$ are the
posterior mean and covariance for state $u$ at time $t=1$, obtained via
Algorithm \ref{alg:ivpsolution}. The posterior over the states is,
\begin{align}
&\left[ u,v, u(t=.5) \mid u^*(t=1), v^*(t=.5) \right] \label
{eqn:MBVposterior} \\
&\propto \left[u, v \mid u(t=.5), u^*(t=1), v^*(t=1) \right] \left[
u(t=.5) \mid u^*(t=1) , v^*(t=.5) \right]
\left[ u(t=.5) \right]. \nonumber
\end{align}
The possibility of multimodality suggests the use of an appropriate
Markov chain Monte Carlo scheme, such as parallel tempering \citep
{Geyer1991}, implemented in Supplement B, 
which can
quickly identify and explore disjoint regions of high posterior probability.

\subsection{Discretization uncertainty for a model of fluid dynamics}
\label{sec:Nav-Stokes}

The following example illustrates that the probabilistic framework
developed may be reliably and straightforwardly applied to a
high-dimensional dynamical system. The Navier--Stokes system is a
fundamental model of fluid dynamics, incorporating laws of conservation
of mass, energy and linear momentum for an incompressible fluid over a
domain given constraints imposed along the boundaries. It is an
important component of complex models in oceanography, weather,
atmospheric pollution, and glacier movement. Despite its extensive use,
the dynamics of Navier--Stokes models are poorly understood even at
small time scales, where they can exhibit turbulence.

\begin{figure}[ht!]
\includegraphics[scale=0.95]{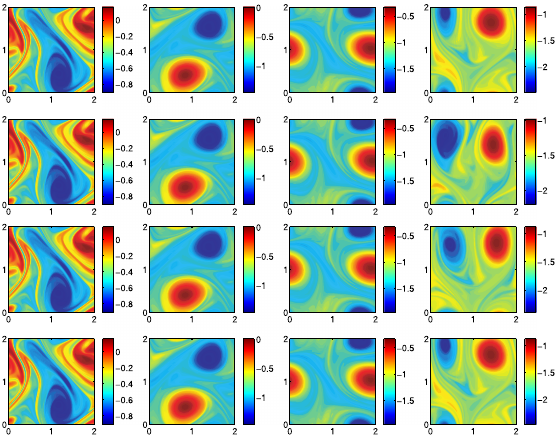}\vspace*{-3pt}
\caption{Time evolution of four forward simulated realizations (along
rows) of fluid vorticity,
governed by the forced Navier--Stokes model (\ref
{fig:NavierStokesPDE}), over two spatial dimensions:
the angle of the inner ring (horizontal axis) and outer ring (vertical
axis) of a two dimensional
torus. Angles are expressed in radians. Vorticities are evaluated at
times ${\bf t} = (0.2, 0.4, 0.6, 0.8)$
units (along columns).}
\label{fig:navstokes}\vspace*{-3pt}
\end{figure}

The Navier--Stokes PDE model for the time evolution of 2 components of
the velocity, $u:\mathcal{D}\to\mathbb{R}^2$, of an incompressible
fluid on a torus, $\mathcal{D} = [0,2\pi)\times[0,2\pi]$, can be
expressed in spherical coordinates as,
\begin{align}
\left\{
\begin{array}{l l l l }
u_t - \theta\,\Delta u + \left(u \cdot\nabla\right) u & = &
\textsf{g} - \nabla\textsf{p}, \quad\quad& (x,t) \in\mathcal
{D}\times[a,b],\\
\nabla\cdot u & = & 0, & (x,t) \in\mathcal{D}\times[a,b],\\
\int u^{(j)}\,\mbox{d}x & = & 0, & (x,t) \in\mathcal{D}\times[a,b],
\; j = 1,2,\\
u & = & u^*, & (x,t) \in\mathcal{D}\times\{0\},
\end{array}
\right.
\label{fig:NavierStokesPDE}
\end{align}
\noindent where $\Delta:=[\tfrac{\partial^2}{\partial
x^2_1}+\tfrac{\partial^2}{\partial x^2_2}]$ is the Laplacian
operator such that $\Delta u = (u^{(1)}_{x_1 x_1} + u^{(1)}_{x_2
x_2}, u^{(2)}_{x_1 x_1}{+}\break u^{(2)}_{x_2 x_2})$, and $ \nabla
:=[\tfrac{\partial}{\partial x_1}+\tfrac{\partial}{\partial
x_2}]$ is the gradient operator such that $\nabla u =(
u^{(1)}_{x_1} + u^{(1)}_{x_2},\break u^{(2)}_{x_1} +u^{(2)}_{x_2} )$.
The model is parametrized by the viscosity of the fluid, $\theta\in
\mathbb{R}^+$, the pressure function $\textsf{p}: \mathcal{D}\times
[a,b] \to\mathbb{R}$, and the external time-homogeneous forcing
function $\textsf{g} := 5\times10^{-3} \cos[(\tfrac
{1}{2},\tfrac{1}{2}) \cdot x ]$. We further assume
periodic boundary conditions, and viscosity $\theta= 10^{-4}$ in the
turbulent regime. The exact solution of the Navier--Stokes boundary
value problem (\ref{fig:NavierStokesPDE}) is not known in closed form.

We will visualize the probabilistic solution of the Navier--Stokes
system by reducing the two components of velocity to a one dimensional
function: the local spinning motion of the incompressible fluid, called
vorticity, which we define as, $\varpi= -\nabla\times u$, where
$\nabla\times u$ represents the rotational curl (the cross product of
$\nabla$ and $u$), with positive vorticity corresponding to clockwise
rotation. We discretize the Navier--Stokes model (\ref
{fig:NavierStokesPDE}) over a grid of size $128$ in each spatial
dimension. A pseudo spectral projection in Fourier space yields 16,384
coupled, stiff ODEs with associated constraints (details provided on
the accompanying website).
Figure \ref{fig:navstokes} shows four forward simulated vorticity
trajectories (along rows), obtained from two components of velocity at
four distinct time points (along columns). Differences in the state
dynamics arise by the final time shown, where the four trajectories
visibly diverge from one another. These differences describe
uncertainty resulting from discretizing the exact but unknown infinite
dimensional solution.

\subsection{Discretization uncertainty for partial differential
equations}\label{sec:HeatPDE}

Similarly to the ODE case, we model the prior belief about the fixed
but unknown solution and partial derivatives of a PDE boundary value
problem by defining a probability measure over multivariate
trajectories, as well as first or higher order derivatives with respect
to spatial inputs. The prior structure will depend on the form of the
PDE model under consideration. We illustrate discretization uncertainty
modelling for PDEs by considering the parabolic heat equation
describing heat diffusion over time along a single spatial dimension
implicitly via,
\begin{align}
\left\{
\begin{array}{l l l l l}
u_t(x, t) &=& \kappa u_{xx}(x, t), \quad\quad& t\in[0,0.25], \; & x
\in[0,1],\\
u(x, t) &=& \sin\left( x \pi\right), \quad\quad& t=0, \; & x \in
[0,1],\\
u(x, t) &=& 0, \quad\quad& t\in(0,0.25], \; & x \in\{0,1\},
\end{array}
\right.
\label{eqn:HeatPDE}
\end{align}
with conductivity parameter of $\kappa= 1$. Note that this PDE has the
explicit form (\ref{eqn:vecfield}) with $\mathcal{A}=\{u, u_t\}$ and
$\mathcal{B}=u_{xx}$.

One modelling choice for extending the prior covariance to a
spatio-temporal domain is to assume separability by defining
covariances over time and space independently. Accordingly we choose
temporal kernel $ R_{\lambda}$ as in Section \ref{sec:ivpprior} and
let $ R_{\nu}: \mathbb{R}\times\mathbb{R}\to\mathbb{R}$ be a
possibly different spatial kernel function with length-scale $\nu$. We
integrate the spatial kernel once to obtain $ Q_\nu(x_j,x_k) =\int
_{0}^{x_j} R_\nu(z,x_k)\mbox{d}z$ and twice to obtain $S_\nu
(x_j,x_k) =\int_{0}^{x_j} Q_\nu(z,x_k)\mbox{d}z$. The covariance
structures between temporal and spatial evaluation points $[x_j, t_j]$
and $[x_k, t_k]$ are obtained from the process convolution formulation
with temporal and spatial prior precision parameters $\alpha$ and
$\beta$ respectively, as follows,\looseness=1
\begin{align*}
\begin{array}{lll}
C^0_{xx}([x_j,t_j],[x_k,t_k]) &=& \alpha^{-1}\beta^{-1}\smallint
_{\mathbb{R}} R_{\nu}(x_j, z) R_{\nu}(x_k, z) \mathrm{d}z \smallint
_{\mathbb{R}} Q_\lambda(t_j, w) Q_\lambda(t_k, w) \mathrm{d}w, \\
C_t^0([x_j,t_j],[x_k,t_k]) &=& \alpha^{-1}\beta^{-1}\smallint
_{\mathbb{R}} S_{\nu}(x_j, z) S_{\nu}(x_k, z) \mathrm{d}z \smallint
_{\mathbb{R}} R_\lambda(t_j, w) R_\lambda(t_k, w) \mathrm{d}w, \\
C^0([x_j,t_j],[x_k,t_k]) &=& \alpha^{-1}\beta^{-1}\smallint_{\mathbb
{R}} S_{\nu}(x_j, z) S_{\nu}(x_k, z) \mathrm{d}z \smallint_{\mathbb
{R}} Q_\lambda(t_j, w) Q_\lambda(t_k, w) \mathrm{d}w,
\end{array}
\end{align*}
with cross-covariances defined analogously. The mean function and
covariance specification uniquely determine the joint GP prior on the
state and its partial derivatives given hyperparameters $\Psi$, which
include $\alpha,\beta,\lambda,m^0_{xx}, m^0_t,m^0$, at vectors of
temporal and spatio-temporal locations $({\bf x_j}, {\bf t_j})$, $({\bf
x_k}, {\bf t_k})$, $({\bf x_l}, {\bf t_l})$ as:
\begin{align}
\left[
\begin{matrix} u_{xx}({\bf x_j},{\bf t_j}) \\ u_t({\bf x_k},{\bf
t_k})\\ u({\bf x_l},{\bf t_l})
\end{matrix}
\right] \sim
{\cal{GP}}\left(\left[
\begin{matrix} m^0_{xx}({\bf x_j},{\bf t_j}) \\ m^0_t({\bf x_k},{\bf
t_k})\\ m^0({\bf x_l},{\bf t_l})
\end{matrix}
\right]
, \left[
\begin{matrix} C_{1,1} & C_{1,2} & C_{1,3}\\
C_{2,1} & C_{2,2} & C_{2,3} \\
C_{3,1} & C_{3,2} & C_{3,3}
\end{matrix}
\right] \right),
\label{eqn:jointpriorheateqn}
\end{align}
where,
\begin{align*}
\begin{array}{l@{\ }ll@{\ }l}
C_{1,1} & = C^0_{xx}([{\bf x_j},{\bf t_j}],[{\bf x_j},{\bf t_j}]),
\quad&
C_{1,2} & = \smallint_{0}^{{\bf x_k}}\smallint_{0}^{{\bf z}} \tfrac
{\partial}{\partial{\bf t_k}}C^0_{xx}([{\bf x_j},{\bf t_j}],[{\bf
w},{\bf t_k}]) \mbox{d}{\bf w}\mbox{d}{\bf z},\\
C_{1,3} & = \smallint_{0}^{{\bf x_l}}\smallint_{0}^{{\bf z}}
C^0_{xx}([{\bf x_j},{\bf t_j}],[{\bf w},{\bf t_l}]) \mbox{d}{\bf
w}\mbox{d}{\bf z}, &
C_{2,2} & = C_t^0([{\bf x_k},{\bf t_k}],[{\bf x_k},{\bf t_k}]),\\
C_{2,3} & = \smallint_{0}^{{\bf t_l}} C_t^0([{\bf x_k},{\bf
t_k}],[{\bf x_l},{\bf w}])\mbox{d}{\bf w}, \quad&
C_{3,3} & = C^0([{\bf x_l},{\bf t_l}],[{\bf x_l},{\bf t_l}]).
\end{array}
\end{align*}
When integrated covariances are available in closed form (e.g., Supplement D.4, 
the above matrices can be computed
analytically.

\begin{figure}
\includegraphics[scale=1.05]{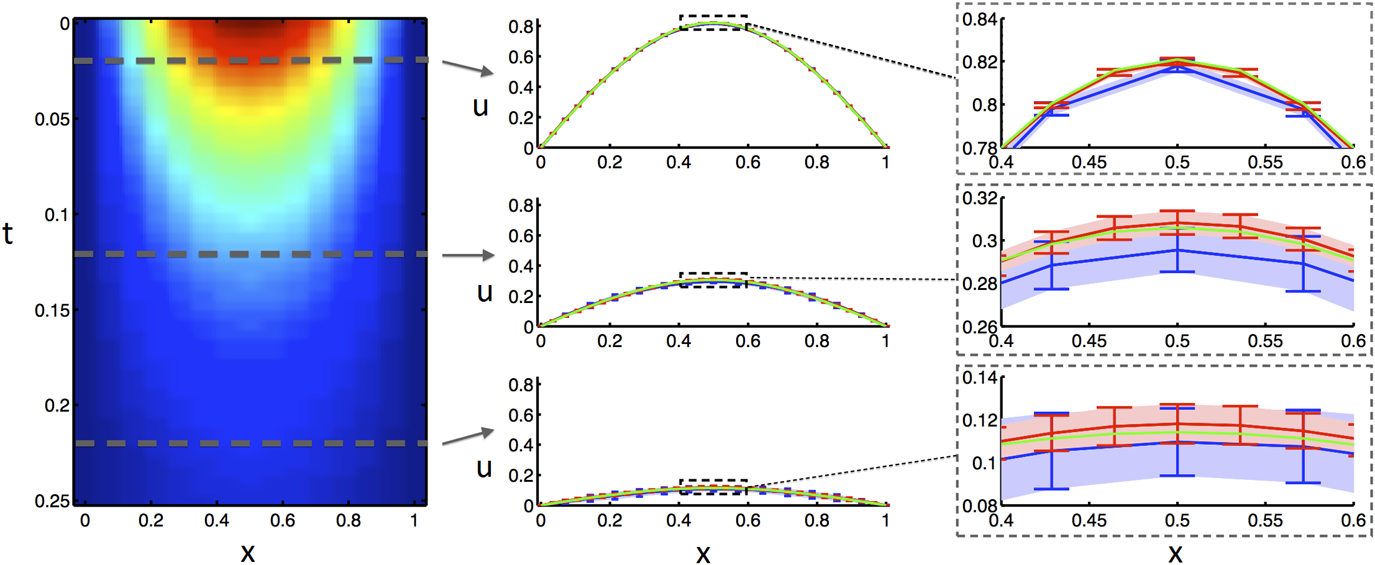}
\caption{Uncertainty in the solution of the heat equation discretized
using two spatio-temporal
grids in blue and red respectively: $15 \times 50$ and $29 \times 100$. Spatial posterior
predictions are shown at $t=(0.02, 0.12, 0.22)$. The exact solution is
shown in green; error
bars show the empirical mean and 2 standard deviations computed from 50
simulations.}
\label{fig:Heat_PDE_Mesh_Comparison}
\end{figure}

A generalization of the sequential updating scheme presented in
Algorithm \ref{alg:ivpsolution} is outlined in Supplement B. 
The \textit{forward in time, continuous in space}
(FTCS) sampling scheme takes steps forward in time, while smoothing
over spatial model interrogations at each step. In Figure \ref
{fig:Heat_PDE_Mesh_Comparison} the exact solution for the PDE initial value
problem (\ref{eqn:HeatPDE}) with $\kappa=1$ on the domain
$[0,1]\times[0,0.25]$ is compared with the probability model of
uncertainty given two different grid sizes. The simulations illustrate
that as the mesh size becomes finer, the uncertainty in the solution decreases.

The characterization of spatial uncertainty is critical for more
complex models in cases where there are computational constraints
limiting the number of system evaluations that may be performed. We
illustrate its effect by performing posterior inference for the
parameter $\kappa$ from data simulated from an exact solution with
$\kappa= 1$ and zero mean Gaussian error with standard deviation of
0.005. Figure \ref{fig:Heat_PDE_Inference_Comparison} shows the
posterior distribution over $\kappa$ obtained by using both a
deterministic ``forward in time, centred in space'' (FTCS) finite
difference solver and a probabilistic solver under a variety of
discretization grids. For the probabilistic solution, a squared
exponential covariance structure was chosen over the spatial domain,
with length-scale 1.5 times the spatial step size. For computational
efficiency a uniform covariance structure was chosen over the temporal
domain, with length-scale 2 times the temporal step size. The prior
precision was set to 10,000. Note the change in the scale as the mesh
is refined and the posterior variance decreases. The use of a
deterministic solver illustrates the problem of posterior
under-coverage that may occur if discretization uncertainty is not
taken into account and too coarse a grid is employed relative to the
complexity of model dynamics. If the discretization is not fine enough,
the approximate posterior assigns negligible probability mass to the
true value of $\kappa$, in contrast to the unbiased posterior obtained
when employing the exact solution for inference. In this illustrative
setting, the use of a probabilistic solver propagates discretization
uncertainty in the solution through to the posterior distribution over
the parameters. By using a probabilistic solver with the same coarsely
discretized grid in this example, the posterior density is more diffuse
but approximately centred on the true value of $\kappa$.

\begin{figure}
\includegraphics{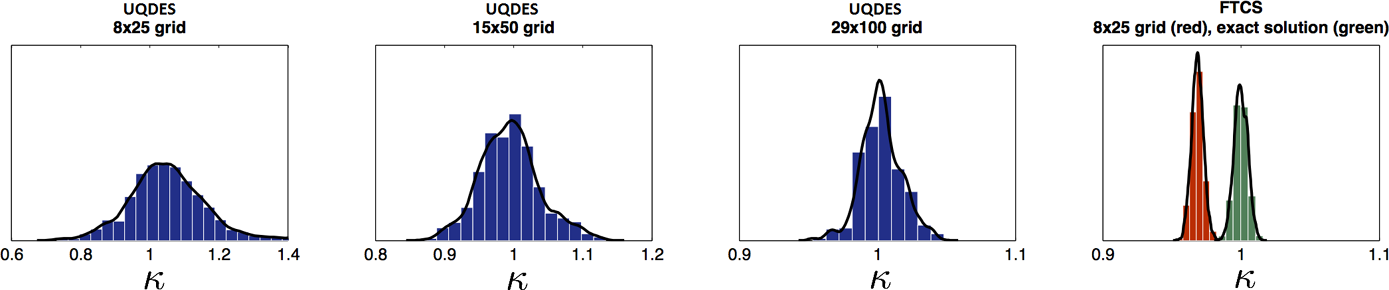}
\caption{Posterior densities for the conductivity parameter in the
heat equation given simulated
data over an $8\times 25$ spatio-temporal grid from the exact solution
with $\kappa= 1$. Inference
is based on the probabilistic differential equation solver using three
grid sizes (blue). The
posterior density based on a deterministic FTCS numerical scheme (red)
and the posterior density
based on the exact solution (red) are also provided.}
\label{fig:Heat_PDE_Inference_Comparison}
\end{figure}

\section{Discussion} \label{sec:discussion}

This paper presents a Bayesian formalism for characterizing the
structure of discretization uncertainty in the solution of differential
equations. Discretization uncertainty for dynamical systems may be
naturally described in terms of degrees of belief, an idea first
formulated by \citet{Skilling1991} and \citet{OHagan1992}. A Bayesian
function estimation approach offers a way of defining and updating our
beliefs about probable model trajectories conditional on interrogations
of the differential equation model. The resulting uncertainty over the
solution can then be propagated through the inferential procedure by
defining an additional layer in the Bayesian posterior hierarchy. As
such, the proposed formulation is a radical departure from the existing
practice of constructing an approximate statistical error model for the
data using numerical techniques.

The more general goal of studying other sources of model uncertainty
requires first understanding how much is known about the dynamics of
the proposed model and how that information changes with time, space,
and parameter settings. Fundamentally, the probabilistic approach we
advocate defines a trade-off between solution accuracy and the size of
the inverse problem, by choice of discretization grid size. This
relationship has been of deep interest in the uncertainty
quantification community \citep[see, e.g.,
][]{KaipioEtAl2004,ArridgeEtAl2006} and our approach offers a way of
resolving the individual contributions to overall uncertainty from
discretization, model misspecification \citep[in the sense
of][]{KennedyOhagan2001}, and Monte Carlo error.

The area of probabilistic numerics is an emerging field at the
intersections of Statistics with Applied Mathematics, Numerical
Analysis, and Computer Science, and we hope that this paper will
provide a starting point for further work in this area. In particular,
while the proposed approach is computationally competitive with first
order numerical methods, it will be important to implement more
efficient sampling algorithms analogous to those of higher order
deterministic solvers developed from over a century of research into
numerical methods. For example a sampling scheme developed based on
multiple model interrogations located between discretization grid
points. Because many traditional higher order numerical solvers (e.g.
explicit fourth order Runge--Kutta) use such a carefully constructed
weighted average of model interrogations, our approach also provides a
way to interpret numerical method in the context of a general
probabilistic framework \citep[e.g.,][]{HenningAndFriends}. It is not
surprising that a self-consistent Bayesian modelling approach has
similarities to numerical methods whose form has been chosen for their
convergence properties. From the perspective presented in this paper,
rather than emulate numerical solvers, quadrature is automatically
provided by the Gaussian process integration that occurs at each step
of Algorithm \ref{alg:ivpsolution}. If the solution smoothness is
correctly specified \textit{a priori} via the covariance structure,
then the Gaussian process integration step will naturally yield an
appropriate quadrature rule.

Modelling uncertainty about an exact but unknown solution may be viewed
as providing both an estimate of the exact solution and its associated
functional error analysis. Although proportionally scalable,
uncertainty quantification methods are naturally more expensive than
deterministic numerical approximations without associated error
analyses, making efficient sampling strategies for stochastic model
calibration an important goal for future research. As may be expected,
low dimensional, stable systems, without any topological restrictions
on the solution may not suffer from significant discretization effects
at a reasonably fine discretization grid. However, it is often
difficult to determine \textit{a priori} whether this stability will
persist across the entire parameter space. An example is illustrated by
the JAK-STAT case study, where the relationships between model
parameters and solver specifications (reflecting changes in the system
dynamics) may introduce bias in the estimation of parameters. If model
uncertainty is indeed negligible and comparable over all parameter
settings, sampled trajectories obtained via the probabilistic method
will have little variability around the exact solution. If, however,
the chosen discretization grid is too coarse with respect to the rate
of change of system dynamics over some subset of the input space, the
uncertainty associated with the model error will be accounted and
propagated. For this reason, it may be prudent to adopt a probabilistic
approach if one cannot be reasonably certain of negligible and
comparable discretization error across all parameter settings.

Importantly, the probabilistic framework can be used to optimize
computation. For a given discretization grid size, a natural question
is how we can arrange the grid on the domain in such a way that the
resulting approximation is as close as possible in terms of some
specified criterion. Most commercially available numerical differential
equation solvers select the step length of the discretization grid
sequentially. Controlling the local error by choice of the distance
between neighbouring grid locations is called \textit{adaptive step
size selection}. For example, in the simplest cases, this is
accomplished by evaluating the local truncation error at each step and
halving the step size if this exceeds an error tolerance that is
pre-specified by the user (the process may be repeated several times
per step until an appropriate local truncation error is achieved). In a
similar way to local truncation error, a predictive probability model
of discretization uncertainty may be used to guide sequential mesh
refinement for complex dynamical systems. Optimizing a chosen design
criterion would yield the desired discretization grid, as suggested by
\cite{Skilling1991}. Some initial work in this direction has been
conducted in \citet{Chkrebtii2013}, where it is shown that the natural
Kullback--Leibler divergence criterion may be successfully used to
adaptively choose the discretization grid in a probabilistic manner.

Prior specification may be used to rule out physically inconsistent
trajectories that may arise from a numerical approximation. In
numerical analysis, specialized techniques are developed and shown to
satisfy certain properties or constraints for a given class of
differential equations. A probabilistic framework would permit imposing
such constraints in a more flexible and general way. Smoothness
constraints on the solution may be incorporated through the prior, and
anisotropy in the trajectory may be incorporated at each step through
the covariance and the error model. Although in such general cases it
may no longer be possible to obtain closed form representations of the
sequential prior updates, an additional layer of Monte Carlo sampling
may be used to marginalize over the realizations required to
interrogate the differential equation model.

These developments lie at the frontier of research in uncertainty
quantification, dealing with massive nonlinear models, complex or even
chaotic dynamics, and strong spatial-geometric effects (e.g. subsurface
flow models). Even solving such systems approximately is a problem that
is still at the edge of current research in numerical
analysis.\looseness=1

Inference and prediction for computer experiments \citep
{SacksEtAl1989} is based on a model of the response surface, or
emulator, constructed from numerical solutions of prohibitively
expensive large-scale system models over a small subset of parameter
regimes. Currently, numerical uncertainty is largely ignored when
emulating a computer model, although it is informally incorporated
through a covariance nugget \citep[see, e.g.,][]{GramacyLee2012} whose
magnitude is often chosen heuristically. Adopting a probabilistic
approach on a large scale will have practical implications in this area
by permitting relaxation of the error-free assumption adopted when
modelling computer code output, leading to more realistic and flexible
emulators. For this reason, Bayesian calibration of stochastic
emulators 
will be an important area of future research.\looseness=1

Numerical approximations of chaotic systems are inherently
deterministic, qualitative dynamics (of models such as Lorenz63) are
often studied statistically by introducing artificial perturbations on
the numerical trajectory. The rate of exponential growth (the \textit
{Lyapunov exponent}) between nearby trajectories can be estimated via
Monte Carlo by introducing small changes to the numerical trajectory
over a grid defined along the domain. Our proposed approach explicitly
models discretization uncertainty and therefore offer an alternative to
such artificial perturbations, thus opening new avenues for exploring
uncertainty in the system's qualitative features.


\begin{supplement}
\stitle{Supplementary Material for ``Bayesian Solution Uncertainty Quantification for Differential Equations''}
\sdatatype{.pdf}
\sfilename{bayesian-solution-uq-for-differential-equations-chkrebtii-et-al-2016-supplement.pdf}
\slink[doi]{10.1214/16-BA1017SUPP}
\end{supplement}

\bibliographystyle{ba}

\begin{acknowledgement}
We would like to thank Tan Bui, Simon Byrne, Mike Dowd, Des Higham,
Richard Lockhart, Youssef Marzouk, Giovanni Petris, Daniel Simpson,
John Skilling, Andrew Stuart, Haoxuan Zhou, Wen Zhou, and anonymous referees for
their invaluable comments and suggestions. We would also like to
acknowledge the funding agencies that have supported this research.
Natural Sciences and Engineering Research Council of Canada supported
O.~A. Chkrebtii (2010:389148) and D. A. Campbell (RGPIN04040-2014). B.
Calderhead gratefully acknowledges his previous Research Fellowship
through the 2020 Science programme, funded by EPSRC grant number
EP/I017909/1 and supported by Microsoft Research, as well as support in
part by EPSRC Mathematics Platform grant EP/I019111/1. M.~A. Girolami is
supported by the UK Engineering and Physical Sciences Council (EPSRC)
via the Established Career Research Fellowship EP/J016934/1 and the
Programme Grant Enabling Quantification of Uncertainty for Large-Scale
Inverse Problems, EP/K034154/1, \url{http://www.warwick.ac.uk/equip}.
He also gratefully acknowledges support from a Royal Society Wolfson
Research Merit Award.
\end{acknowledgement}


\appendix

\section{Brief overview of differential equation models}\label{sec:OverviewOfDEs}


 {\em Ordinary Differential Equation} (ODE) models represent the evolution of system states $u \in\mathcal{U}$ over time $t\in[0,L], L>0$ implicitly through the equation $u_t(t, \theta)= f\left\{ t, u(t,\theta), \theta \right\}$. The standard assumptions about the known function $f$ are that it is continuous in the first argument and Lipschitz continuous in the second argument, ensuring the existence and uniqueness of a solution corresponding to a known initial state $u^*(0)$.  Inputs and boundary constraints define several model variants.  

The {\em Initial Value Problem} (IVP) of order one models time derivatives of system states as known functions of the states constrained to satisfy initial conditions, $u^*(0)$, 
\begin{align}
 \left\{ 
  \begin{array}{l l}
    u_t(t) = f \left\{ t, u(t), \theta \right\}, & \quad t \in [0,L], \, L>0, \\
    {u}(0) = u^*(0).
  \end{array} \right. 
\label{eqn:IVPsystem}
\end{align}
The existence of a solution is guaranteed under mild conditions \citep[see for example,][]{Butcher2008,CoddingtonLevinson1955}.   Such models may be high dimensional and contain other complexities such as algebraic components, functional inputs, or higher order terms.

While IVP models specify a fixed initial condition on the system states, the {\em Mixed Boundary Value Problem} (MBVP) may constrain different states at different time points.  Typically these constraints are imposed at the boundary of the time domain. For example, the two state mixed boundary value problem is,
\begin{align}
 \left\{ 
\begin{array}{lll}
\left[u_t(t), v_t(t)\right] =f \left\{ t,\left[u(t), v(t)\right], \theta \right\}, \quad\quad & t\in[0, L],\, L>0, \\ 
\phi\left\{ v(0), u(L) \right\} = 0,
\end{array}\right.
\label{eqn:multiBVP}
\end{align}
 and can be straightforwardly generalized to higher dimensions and extrapolated beyond the final time point $L$.   Whereas a unique IVP solution exists under relatively mild conditions, imposing mixed boundary constraints can result in multiple solutions \citep[e.g.][]{Keller1968} introducing possible problems for parameter estimation. 

The {\em Delay Initial Function Problem} (DIFP) generalizes the initial constraint of IVPs to an initial function, $\phi(t)$, thereby relating the derivative of a process to both present and past states at lags $\tau_j \in [0, \infty)$,
\begin{align}
 \left\{ 
  \begin{array}{l l l}
    u_t(t) &= f \left\{ t, u(t-\tau_1),\ldots,u(t-\tau_d), \theta \right\} , & \; t \in [0, L],\, L>0, \\
     u(t) &= \phi(t), &  \; t \in \left [0-\max_{j}( \tau_j), 0\right ].
  \end{array} \right. 
\label{eqn:DIFP}
\end{align}
DIFPs are well suited to describing biological and physical dynamics that take time to propagate through systems.  However, they  pose challenges to numerical techniques due to potentially large and structured truncation error   \citep{BellenZennaro2003}.   

{\em Partial Differential Equation} (PDE) models represent the derivative of states with respect to multiple arguments, for example time and spatial variables. PDE models are ubiquitous in the applied sciences, where they are used to study a variety of phenomena, from animal movement, to the propagation pattern of a pollutant in the atmosphere.  The main classes of PDE models are based on elliptic, parabolic and hyperbolic equations.  Further adding to their complexity,  functional boundary constraints and initial conditions make PDE models  more challenging, while their underlying theory is less developed compared to ODE models \citep{PDEBOOK}.

\section{Problem specific algorithms}\label{sec:algorithms}

This section provides algorithms for the extension of the probabilistic modelling of solution uncertainty to deal with the MBVP and DIFP.  Algorithm \ref{alg:mbvpsolution} provides a Metropolis-Hastings implementation of the Markov chain Monte Carlo algorithm from the two state mixed boundary value problem (\ref{eqn:multiBVP}) with boundary constraint $[ v(t=0),u(t=L) ] - [v^*(t=0), u^*(t=L) ] = [0,0]$.  Since these models may be subject to solution multiplicity, Algorithm \ref{alg:mbvppartemp} provides the corresponding parallel tempering \citep{Geyer1991} implementation to permit efficient exploration of the space of $u(t=0)$.   Algorithm \ref{alg:ddesolution} describes forward simulation for the delay initial function problem (\ref{eqn:DIFP}). In the context of the inverse problem, Algorithm \ref{alg:DDEInvMH} describes a Metropolis-Hastings implementation of the Markov chain Monte Carlo sampler for drawing realizations from the posterior distribution of the unknowns in the JAK-STAT system problem described in Section \ref{subsec:jakstatexample}. Algorithm \ref{alg:DDEInvPT} describes its corresponding parallel tempering implementation.   Algorithm \ref{alg:pdesolution} constructs a probabilistic solution for the heat equation PDE boundary value problem (\ref{eqn:HeatPDE}).

\setcounter{algorithm}{2}


\begin{algorithm}[!htp]
\caption{For the ODE mixed boundary value problem  (\ref{eqn:multiBVP}), draw $K$ samples from the forward model conditional on boundary values $ u^*(t = L), v^*(t = 0), \theta, \Psi$, and a discretization grid ${\bf s}=  (s_1,\cdots,s_N)$ }
\label{alg:mbvpsolution}
\begin{algorithmic}
\STATE At time $s_1 = 0$ initialize the unknown boundary value $u(t=0)\sim \pi(\cdot)$ by sampling from the prior $\pi$;
\STATE Use Algorithm \ref{alg:ivpsolution} to conditionally simulate  $(\mathbf{u};\mathbf{v}) =\left(u(t_1),\ldots,u(t_T); v(t_1),\ldots,v(t_T)\right)$ from the probabilistic solution of model (\ref{eqn:IVPsystem}) given ${u}(t=0), {v}^*(t=0), \theta, \Psi$;

\FOR{$k = 1:K$}
\STATE Propose unknown boundary value $u'(t=0) \sim q\left\{\cdot \mid u(t=0)\right\}$ from a proposal density $q$;
\STATE Use Algorithm \ref{alg:ivpsolution} to conditionally simulate the probabilistic solution of model (\ref{eqn:IVPsystem}), $(\mathbf{u}';\mathbf{v}')$, given  $u'(t=0), v^*(t=0), \theta,\Psi$;
\STATE Compute the rejection ratio,
\[
\rho  = \;
\frac{q\left\{ u(t=0) \mid u'(t=0)\right\}}{q\left\{u'(t=0) \mid u(t=0)\right\}}\;
\frac{\pi\left\{u(t=0)\right\}}{\pi\left\{u'(t=0)\right\}}\;
\frac{p \left\{ u(t=L)  \mid u^*(t=L)\right\}}{p\left\{u'(t=L) \mid u^*(t=L) \right\}};
\]
\IF{$\min\{ 1, \rho \}  > \mbox{U}[0,1] $ }
\STATE Update $u(t=0) \leftarrow u'(t=0)$;
\STATE Update $(\mathbf{u};\mathbf{v}) \leftarrow (\mathbf{u}';\mathbf{v}')$;
\ENDIF
\STATE Return $(\mathbf{u};\mathbf{v})$.
\ENDFOR
\end{algorithmic}
\end{algorithm}


\begin{algorithm}[!htp]
\caption{Parallel tempering implementation of Algorithm \ref{alg:mbvpsolution} with $C$ chains.}
\label{alg:mbvppartemp}
\begin{algorithmic} 
\STATE Choose the probability $\xi$ of performing a swap move between two randomly chosen chains at each iteration, and define a temperature vector $\gamma \in (0,1]^C$ , where $\gamma_{C} = 1$;

\STATE At time $s_1 = 0$ initialize the unknown boundary value $u(t=0)_{(1)} \sim \pi(\cdot)$ by sampling from the prior $\pi$ and set $u(t=0)_{(c)} = u(t=0)_{(1)}$  for $c=2,\ldots,C$;
\STATE Use Algorithm \ref{alg:ivpsolution} to conditionally simulate $(\mathbf{u};\mathbf{v})_{(1)}$ from the probabilistic solution of model (\ref{eqn:IVPsystem}) given  $u(t=0)_{(1)}, v^*(t=0), \theta,\Psi$ and set $(\mathbf{u};\mathbf{v})_{(c)}= (\mathbf{u};\mathbf{v})_{(1)}$  for $c=2,\ldots,C$;
\FOR{$k = 1:K$}
\IF{$ \xi  > \mbox{U}[0,1] $  }
\STATE Propose a swap between $ i,j \sim q(i,j)$, $1\leq i,j \leq C, i\neq j$, where $q$ is a proposal.
\STATE Compute the rejection ratio,
\[
\rho = \frac{p \left\{ u(t=L)_{(i)} \mid u^*(t=L)\right\}^{\gamma_i}}{p \left\{u(t=L)_{(i)} \mid u^*(t=L) \right\}^{\gamma_j}} \,
\frac{p \left\{ u(t=L)_{(j)} \mid u^*(t=L)\right\}^{\gamma_j}}{p \left\{ u(t=L)_{(j)} \mid u^*(t=L) \right\}^{\gamma_i}};
\] 
\IF{$\min ( 1, \rho )  > \mbox{U}[0,1] $ }
\STATE Swap initial conditions $u(t=0)_{(i)}$ and $u(t=0)_{(j)}$;
\ENDIF
\ENDIF 
\FOR{$c = 1:C$}
\STATE Perform one iteration of Metropolis-Hastings Algorithm \ref{alg:mbvpsolution}, using instead a tempered likelihood to compute the rejection ratio:
 \[
\rho  = \;
\frac{q\left\{u(t=0)_{(c)} \mid u'(t=0)_{(c)}\right\}}{q\left\{u'(t=0)_{(c)} \mid u(t=0)_{(c)}\right\}}\;
\frac{\pi\left\{u(t=0)_{(c)}\right\}}{\pi\left\{u'(t=0)_{(c)}\right\}}\;
\frac{p \left\{ u(t=L)_{(c)}  \mid u^*(t=L)\right\}^{\gamma_c}}{p \left\{u'(t=L)_{(c)} \mid u^*(t=L) \right\}^{\gamma_c}},
\]
 with temperature $\gamma_c$;
\ENDFOR 
\STATE Return $(\mathbf{u};\mathbf{v})_{(C)}$.
\ENDFOR
\end{algorithmic}
\end{algorithm}


\begin{algorithm}[!htp]
\caption{For the ODE delay initial function problem (\ref{eqn:DIFP}), draw one sample from the forward model over grid ${\bf t} = (t_1,\cdots,t_T)$ given $\phi,\tau,\theta, \Psi,$ and a discretization grid ${\bf s} = (s_1,\cdots,s_N)$}
\label{alg:ddesolution}
\begin{algorithmic}
\STATE
\STATE At time $s_1 = 0$, sample $u^0(0-\tau)\sim \phi(0-\tau)$ and initialize the derivative $\mbox{f}_{1}= f \left\{s_1, u^*(0), u^0(0-\tau), \theta\right\}$ and define $m^0, m^0_t, C^0, C^0_t$ as in Section \ref{sec:priors};
\FOR{$n = 1 : N$}
\STATE If $n=1$, set $g_1=C_t^{0}(s_1,s_1) +  C_t^{0}(s_1-\tau,s_1-\tau)+ 2C_t^{0}(s_1,s_1-\tau)$, otherwise set $g_n=C_t^{n-1}(s_n,s_n) + C_t^{n-1}(s_n-\tau,s_n-\tau)+ 2C_t^{n-1}(s_n,s_n-\tau) + r_{n-1}(s_n)$; 
\STATE Compute for each system component,
 \begin{align*}
m^n({\bf s})
&=  m^{n-1}({\bf s})  +  g_n^{-1}
\smallint_0^{{\bf s}}   C_t^{n-1}(z,s_n) \mbox{d}z \, \left\{\mbox{f}_n -  m_t^{n-1}(s_n) \right\},\\
m^n({\bf s}-\tau)
&=  m^{n-1}({\bf s}-\tau)  +  g_n^{-1}
\smallint_0^{{\bf s}-\tau}   C_t^{n-1}(z,s_n) \mbox{d}z \, \left\{\mbox{f}_n -  m_t^{n-1}(s_n) \right\} ,\\
m_t^n({\bf s})
&=  m_t^{n-1}({\bf s}) +  g_n^{-1}  \;
C_t^{n-1}({\bf s},s_n)\left\{\mbox{f}_n -   m_t^{n-1}(s_n) \right\},\\
C^n({\bf s},{\bf s})  & =    C^{n-1}({\bf s},{\bf s})  -  g_n^{-1} \smallint_0^{{\bf s}}  C_t^{n-1}(z,s_n)\mbox{d}z \,  \left\{\smallint_0^{{\bf s}} C_t^{n-1}(z,s_n)\mbox{d}z \right\}^\top,\\
 C_t^n({\bf s},{\bf s})  & =  C_t^{n-1}({\bf s},{\bf s})  -   g_n^{-1}  C_t^{n-1}({\bf s},s_n)\, C_t^{n-1}(s_n,{\bf s}),\\
 \smallint_0^{{\bf s}}  C_t^n(z,{\bf s}) \mbox{d}z  & =  \smallint_0^{{\bf s}}C_t^{n-1}(z,{\bf s}) \mbox{d}z -   g_n^{-1}  \smallint_0^{{\bf s}} C_t^{n-1}(z,s_n) \mbox{d}z \, C_t^{n-1}(s_n,{\bf s}),\\
 \smallint_0^{{\bf s}-\tau}  C_t^n(z,{\bf s}) \mbox{d}z  & =  \smallint_0^{{\bf s}-\tau}C_t^{n-1}(z,{\bf s}) \mbox{d}z -   g_n^{-1}  \smallint_0^{{\bf s}-\tau} C_t^{n-1}(z,s_n) \mbox{d}z \, C_t^{n-1}(s_n,{\bf s});
 \end{align*}

\IF {$n < N$}
\STATE Sample step ahead realization $u^{n}(s_{n+1})$ from the predictive distribution of the state,
\[p\left\{ u(s_{n+1})  \mid \mbox{f}_n, \Psi \right\}
 = \mathcal{N} \left\{  u(s_{n+1})  \mid  m^n(s_{n+1}), C^n(s_{n+1},s_{n+1}) \right\};\]
 
\IF{$s_{n+1} - \tau < 0$}
\STATE {Sample $u^n(s_{n+1}-\tau)  \sim \phi(s_{n+1} - \tau)$},
\ELSE 
\STATE {Sample $u^n(s_{n+1}-\tau)  \sim \mathcal{N} \left\{ m^n(s_{n+1}-\tau), C^n(s_{n+1}-\tau,s_{n+1}-\tau) \right\};$}
\ENDIF 

\STATE Interrogate the model by computing $\mbox{f}_{n+1} ={f}\left\{s_{n+1}, u^{n}(s_{n+1}),u^{n}(s_{n+1}-\tau), \theta \right\}$;

\ENDIF
\ENDFOR
\STATE Return 
 $\mathbf{u} = \left(u(t_1),\cdots,u(t_T)\right) \sim \mathcal{GP}\left\{m^N({\bf t}), C^N({\bf t},{\bf t})\right\}$, where ${\bf t}\subset{\bf s}$.
\end{algorithmic}
\end{algorithm}

 
\begin{algorithm}[!htp]
\caption{Draw $K$ samples from (\ref{eqn:posterior_solution}) with density $p \left\{ \theta, \tau, \phi, \alpha,\lambda, \mathbf{u}  \mid  \mathbf{y}, \Sigma\right\} $ given observations of the transformed solution of delay initial function problem (\ref{eqn:DIFP}).}
\label{alg:DDEInvMH}
\begin{algorithmic} 
\STATE
\STATE Initialize $\theta, \tau, \phi, \alpha, \lambda \sim \pi(\cdot)$ where $\pi$ is the prior density;
\STATE Use Algorithm \ref{alg:ddesolution} to conditionally simulate a realization, $\mathbf{u} = \left(u(t_1),\ldots,u(t_T)\right)$, from the probabilistic solution of ODE delay initial function problem  (\ref{eqn:DIFP}) given $\theta, \tau, \phi,$  and discretization grid ${\bf s}= \left(s_1,\cdots,s_N\right)$;
\FOR{$k = 1:K$}
\STATE Propose $\theta',  \tau', \phi' \sim q(\cdot \mid  \theta, \tau, \phi)$ where $q$ is a proposal density;
\STATE Use Algorithm \ref{alg:ddesolution} to conditionally simulate a vector of realizations, $\mathbf{u}'$, from the forward model given $\theta', \tau', \phi', \alpha, \lambda$; 
\STATE Compute the rejection ratio,
\begin{align*}
& \rho  =
\frac{q(\, \theta\, , \tau , \phi \mid  \theta' , \tau', \phi' \,)}{q(\, \theta'\, , \tau' , \phi' \mid  \theta, \tau, \phi \,)} \;
 \frac{\pi( \,\theta, \tau, \phi\,)}{\pi(\, \theta'\, , \tau' , \phi'\,)} \;
 \frac{p\left(\, \mathbf{y} \mid \mathbf{u}, \theta , \Sigma \,\right)}{p \left(\, \mathbf{y} \mid \mathbf{u}', \theta', \Sigma \,\right)};
\end{align*}
\IF{$\min ( 1, \rho )  > \mbox{U}[0,1] $ }
\STATE Update $\left(\theta, \tau, \phi \right) \leftarrow \left( \theta'\, , \tau' , \phi'\right)$;
\STATE Update $\mathbf{u}\leftarrow \mathbf{u}'$;
\ENDIF
\STATE Propose $\alpha', \lambda' \sim q(\cdot \mid \alpha, \lambda )$ where $q$ is a proposal density;
\STATE Use Algorithm \ref{alg:ddesolution} to conditionally simulate a vector of realizations, $\mathbf{u}'$, of the associated probabilistic solution for the ODE delay initial function problem given $\theta, \tau, \phi, \alpha',\lambda',$  and discretization grid ${\bf s}=  \left(s_1,\cdots,s_N\right)$;

\STATE Compute the rejection ratio,
\begin{align*}
& \rho  =
\frac{q( \alpha, \lambda \mid  \alpha', \lambda' )}{q( \alpha', \lambda' \mid \alpha, \lambda)}
 \frac{\pi( \alpha, \lambda )}{\pi( \alpha', \lambda' )}
 \frac{p\left(\, \mathbf{y} \mid \mathbf{u}, \theta, \Sigma \,\right)}{p \left(\, \mathbf{y} \mid \mathbf{u}', \theta, \Sigma \,\right)};
\end{align*}
\IF{$\min ( 1, \rho )  > \mbox{U}[0,1] $ }
\STATE Update $(\alpha, \lambda) \leftarrow (\alpha', \lambda')$;
\STATE Update $\mathbf{u} \leftarrow \mathbf{u}'$;
\ENDIF
\STATE Return $( \theta, \tau, \phi, \alpha, \lambda, \mathbf{u})$.
\ENDFOR
\end{algorithmic}
\end{algorithm}


\begin{algorithm}[!htp]
\caption{Parallel tempering implementation of Algorithm \ref{alg:DDEInvMH}}
\label{alg:DDEInvPT}
\begin{algorithmic} 
\STATE
\STATE Define probability $\xi$ of performing a swap move between two randomly chosen chains at each iteration, and define a temperature vector $\gamma \in (0,1]^C$, where $\gamma_{C} = 1$;
\STATE Initialize $\left(\theta, \tau, \phi, \alpha, \lambda\right)_{(1)} \sim \pi(\cdot)$ where $\pi$ is the prior density and set $\left(\theta, \tau, \phi, \alpha, \lambda\right)_{(c)} = \left(\theta, \tau, \phi, \alpha, \lambda\right)_{(1)}$ for $c = 2,\ldots,C$;
\STATE Use Algorithm  \ref{alg:ddesolution} to conditionally simulate a realization, $\mathbf{u}_{(1)} = \left(u(t_1),\ldots,u(t_T)\right)_{(1)}$, from the forward model and set $\mathbf{u}_{(c)} =\mathbf{u}_{(1)}$ for $c=2,\ldots,C$;
\FOR{$k = 1:K$}
\IF{$ \xi  > \mbox{U}[0,1] $ } 
\STATE Propose a swap between $ i,j \sim q(i,j)$, $1\leq i,j \leq C, i\neq j$, where $q$ is a proposal.
\STATE Compute the rejection ratio,
\[\rho =  \frac{p \left(\, \mathbf{y} \mid \mathbf{u}_{(i)},\theta_{(i)}, \Sigma \,\right)^{\gamma_i}}{p\left(\, \mathbf{y} \mid   \mathbf{u}_{(i)}, \theta_{(i)}, \Sigma \,\right)^{\gamma_j}}\cdot \frac{p \left(\, \mathbf{y} \mid  \mathbf{u}_{(j)},\theta_{(j)}, \Sigma \,\right)^{\gamma_j}}{p\left(\, \mathbf{y} \mid  \mathbf{u}_{(j)}, \theta_{(j)}, \Sigma \,\right)^{\gamma_i}};\]

\IF{$\min ( 1, \rho )  > \mbox{U}[0,1] $ }
\STATE Swap parameter vectors $\left(\theta, \tau, \phi, \alpha, \lambda\right)_{(i)} \leftrightarrow \left(\theta, \tau, \phi, \alpha, \lambda\right)_{(j)}$;
\ENDIF\\
\ENDIF \\\
\FOR{$c = 1:C$}
\STATE Perform one iteration of Metropolis-Hastings Algorithm \ref{alg:DDEInvMH}, using instead a tempered likelihood to compute the rejection ratios:
\[ \rho  =
\frac{q(\, \theta_{(c)} \, , \tau_{(c)} , \phi_{(c)} \mid  \theta_{(c)}' , \tau_{(c)}', \phi_{(c)}' \,)}{q(\, \theta_{(c)}'\, , \tau_{(c)}' , \phi_{(c)}' \mid  \theta_{(c)}, \tau_{(c)}, \phi_{(c)} \,)} \;
 \frac{\pi( \,\theta_{(c)}, \tau_{(c)}, \phi_{(c)}\,)}{\pi(\, \theta_{(c)}'\, , \tau_{(c)}' , \phi_{(c)}'\,)} 
 \cdot  \frac{p\left(\, \mathbf{y} \mid  \mathbf{u}_{(c)}, \theta_{(c)}, \Sigma \,\right)^{\gamma_c}}{p \left(\, \mathbf{y} \mid \mathbf{u}_{(c)}', \theta_{(c)}', \Sigma \,\right)^{\gamma_c}},\]
and,
\[ \rho  =
\frac{q( \alpha_{(c)}, \lambda_{(c)} \mid  \alpha_{(c)}', \lambda_{(c)}' )}{q( \alpha_{(c)}', \lambda_{(c)}' \mid \alpha_{(c)}, \lambda_{(c)})}
 \frac{\pi( \alpha_{(c)}, \lambda_{(c)} )}{\pi( \alpha_{(c)}', \lambda_{(c)}' )}
\frac{p\left(\, \mathbf{y} \mid  \mathbf{u}_{(c)}, \theta_{(c)}, \Sigma \,\right)^{\gamma_{c}}}{p \left(\, \mathbf{y} \mid \mathbf{u}_{(c)}',\theta_{(c)}, \Sigma \,\right)^{\gamma_{c}}},\]
with temperature $\gamma_c$;
\ENDFOR \\
\STATE Return $\left(\theta, \tau, \phi, \alpha, \lambda, \mathbf{u}\right)_{(C)}$.\\
\ENDFOR
\end{algorithmic}
\end{algorithm}


\begin{algorithm}[htp!]
\caption{Sampling from the formward model evaluated over the grid ${\bf X}\otimes{\bf T}$ 
 for the heat equation (\ref{eqn:HeatPDE}), given $\kappa, \Psi, N,M$. }
\label{alg:pdesolution}
\begin{algorithmic}  
\STATE Define temporal discretization grid, ${\bf s} = (s_1,\cdots,s_N)$, and for each $s_n$ define spatial discretization grid   ${\bf z}_n = (z_{1,n},\cdots,z_{M,n}), 1 \leq n \leq N$, and construct the $N\times M$ design matrix ${\bf Z} = ({\bf z}_1; \cdots ; {\bf z}_N)$;

\STATE At time $s_1 = 0$ initialize the second spatial derivative using the boundary function and compute the temporal derivative ${\bf f}_1 = \kappa\, [u_{xx}(z_{1,1}, s_1),\cdots,u_{xx}(z_{M,1}, s_1)]$;

\STATE Define the prior covariance as in Section \ref{sec:HeatPDE};

\FOR{$n = 1 : N$}
\STATE If $n=1$, set $G_1=C_{t}^{0}([{\bf z_1},s_1],[{\bf z_1},s_1])$, otherwise set $G_n=C_{t}^{n-1}([{\bf z}_n,s_n],[{\bf z}_n,s_n]) + r_{n-1}\left([{\bf z}_n,s_n]\right)$; 

\STATE Compute,
 \begin{align*}
m_{xx}^n({\bf Z},{\bf s})
&=  m_{xx}^{n-1}({\bf Z},{\bf s}) +  \tfrac{\partial^2}{\partial {\bf Z}^2}\smallint_0^{{\bf s}} C_{t}^{n-1}([{\bf Z},t],[{\bf z}_n,s_{n}])\mbox{d}t \, G_n^{-1} \left\{{\bf f}_n -  m_{t}^{n-1}({\bf z}_n,s_n) \right\},\\
m_{t}^n({\bf Z},{\bf s})
&=  m_{t}^{n-1}({\bf Z},{\bf s}) +   C_{t}^{n-1}([{\bf Z},{\bf s}],[{\bf z}_n,s_{n}]) \, G_n^{-1} \left\{{\bf f}_n -  m_{t}^{n-1}({\bf z}_n,s_n) \right\},\\
  C_{xx}^{n}([{\bf Z},{\bf s}],[{\bf Z},{\bf s}])  
 & =  C_{xx}^{n-1}([{\bf Z},{\bf s}],[{\bf Z},{\bf s}])  \\
 &  - \left\{\tfrac{\partial^2}{\partial {\bf Z}^2}\smallint_0^{{\bf s}}  C_{t}^{n-1}([{\bf Z},t],[{\bf z}_n,s_n]) \mbox{d}t\right\}\, G_n^{-1} \left\{ \tfrac{\partial^2}{\partial {\bf Z}^2}\smallint_0^{{\bf s}}  C_{t}^{n-1}([{\bf Z},t],[{\bf z}_n,s_n]) \mbox{d}t\right\}^\top,\\
 \tfrac{\partial^2}{\partial {\bf Z}^2}\smallint_0^{{\bf s}} C_{t}^{n}([{\bf Z},t],[{\bf Z},{\bf s}]) \mbox{d}t  
& = \tfrac{\partial^2}{\partial {\bf Z}^2}\smallint_0^{{\bf s}} C_{t}^{n-1}([{\bf Z},t],[{\bf Z},{\bf s}]) \mbox{d}t  \\
&     - \left\{\tfrac{\partial^2}{\partial {\bf Z}^2}\smallint_0^{{\bf s}} C_{t}^{n-1}([{\bf Z},t],[{\bf z}_n, s_n]) \right\} \mbox{d}t \, G_n^{-1}  C_{t}^{n-1}([{\bf z}_n,s_n],[{\bf Z},{\bf s}]),\\
 C_{t}^{n}([{\bf Z},{\bf s}],[{\bf Z},{\bf s}])  
& =  C_{t}^{n-1}([{\bf Z},{\bf s}],[{\bf Z},{\bf s}])  -     C_{t}^{n-1}([{\bf Z},{\bf s}],[{\bf z}_n,s_n]) \, G_n^{-1} C_{t}^{n-1}([{\bf z}_n,s_n],[{\bf Z},{\bf s}]);
 \end{align*}
 
\IF { $n < N$}
\STATE Sample one-time-step-ahead realization of the second spatial derivative of the state, $u_{xx}({\bf z}_n,s_{n+1})$ from the predictive distribution,
\begin{align*}
&p\left( [ u_{xx}(z_{1,n},s_{n+1}),\cdots, u_{xx}(z_{M,n},s_{n+1})]  \mid {\bf f}_n, \kappa, \Psi \right)\\
&= \mathcal{N} \left(  [ u_{xx}(z_{1,n},s_{n+1}),\cdots, u_{xx}(z_{M,n},s_{n+1})]  \mid  m_{xx}^{n}({\bf z}_n,s_{n+1}), C_{xx}^{n}([{\bf z}_n, s_{n+1}], [{\bf z}_n,s_{n+1}]) \right);
\end{align*} 
and interrogate the PDE model by computing   ${\bf f}_{n+1} = \kappa\, [u_{xx}(z_{1,n+1}, s_{n+1}),\cdots,u_{xx}(z_{M,n+1}, s_{n+1})]$ 
\ENDIF
\ENDFOR 
\STATE Return 
$\mathbf{U} = \begin{bmatrix}
u(x_{1,1},t_{1}) & \cdots & u(x_{X,1},t_{1})\\
\vdots & \vdots & \vdots\\
u(x_{1,1},t_{T}) & \cdots & u(x_{X,T},t_{T})\\
\end{bmatrix}
\sim \mathcal{GP}\left(m^N({\bf X},{\bf T}), C^N([{\bf X},{\bf T}],[{\bf X},{\bf T}])\right),$ where ${\bf X}\subset{\bf Z}$ and ${\bf T}_{\cdot,1}\subset{\bf s}$.
\end{algorithmic}
\end{algorithm}

\newpage

\section{Computational Comparison}
\label{sec:compcomparison}

Computational comparison between a one-step-ahead probabilistic algorithm (UQDES-1) and the one-step explicit Euler method (Euler) are provided under difference discretization grids and regimes for the initial value ODE problem (\ref{eqn:toyIVPsystem}).  As discussed in the paper, a probabilistic formalism can yield higher order algorithms by varying the sampling scheme.  Therefore, comparison to higher-order numerical methods, such as Runge-Kutta, were outside the scope of our paper.   Figure \ref{fig:timecomparison} shows the computation times as well as integrated log mean square error (IMSE) for each choice of solver.  The improvement in log IMSE under appropriate length-scale specification illustrates the advantage of the probabilistic method within the inverse problem, where this hyperparameter can be estimated.

\setcounter{figure}{9}

\begin{figure}
\centering
\includegraphics[width=\textwidth]{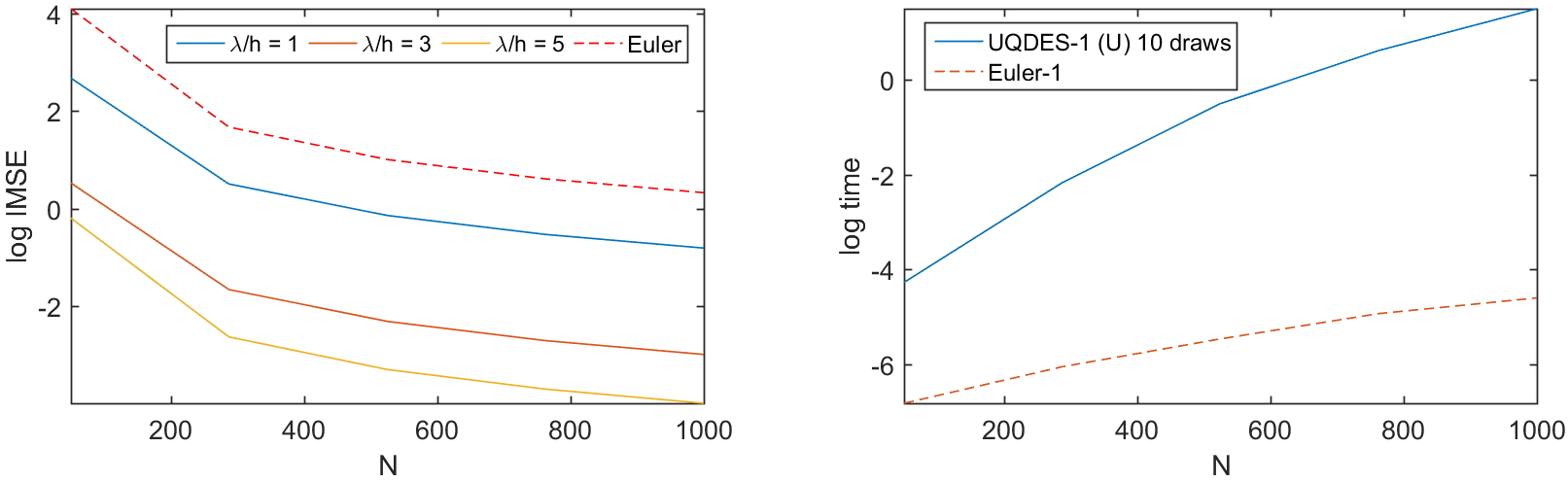}
\caption{\label{fig:timecomparison} Left: log integrated mean squared error is shown for the one-step Euler method (dashed line) and for the mean one-step-ahead probabilistic solution (solid lines) against the discretization grid size (N) for different length-scales expressed as proportions of the time step.  Right: log computational time against grid size is shown for one iteration of the one-step Euler method (red) and 100 iterations of the probabilistic one-step-ahead solver with squared exponential covariance (blue).   }
\end{figure}

\section{Derivations and Proofs}

\subsection{Probabilistic solution as latent function estimation}\label{sec:latentupdate}

We take a process convolution view of the model presented in Section \ref{sec:probalgorithm}.  Let $\mathcal{F}$ 
 be the space of square-integrable random functions and $\mathcal{F}^*$ be its dual space of linear functionals.
The solution and its derivative will be modelled by an integral transform using the linear continuous operators $ R $ and $ Q $ defining a mapping from $\mathcal{F}$ to $\mathcal{F}^*$.  The associated kernels are the deterministic, square-integrable function $R_\lambda: \mathbb{R}\times\mathbb{R} \to \mathbb{R}$ and its integrated version $Q_\lambda(t_j,t_k) = \smallint_0^{t_k} R_\lambda(z,t_k)\mbox{d}z$.
The operators $ R ,  Q ,  R ^\dagger,  Q ^\dagger$ are defined, for $u\in\mathcal{F}$ and $v\in\mathcal{F}^*$, as $ R u(t) = \smallint R_\lambda(t,z) u(z) \mathrm{d}z$ and   $ Q u(t) = \smallint Q_\lambda(t,z) u(z) \mathrm{d}z$, with adjoints $ R ^\dagger v(t) = \smallint R_\lambda(z,t) v(z) \mathrm{d}z$ and  $ Q ^\dagger v(t) = \smallint Q_\lambda(z,t) v(z) \mathrm{d}z$ respectively.

Estimating the solution of a differential equation model may be restated as a problem of inferring an underlying latent process $\zeta\in\mathcal{F}$ from noisy realizations.  We consider the white noise process,
$\zeta \sim \mathcal{N}(0,K)$,  with covariance $K(t_j,t_k)=\alpha^{-1}\delta_{t_j}(t_k)$.  
We model the derivative of the solution as the integral transform,
\begin{align}
 u_t(t_j) =  m_t^n(t_j) +  R  \zeta(t_j), \quad t_j\in [0, L],\,L>0, \, 0 \leq n \leq N.
\label{eqn:derivconv}
\end{align}

\noindent The differential equation solution model at time $t_j \leq b$ is then obtained by integrating the derivative $u_t$ with respect to time over the interval $[0,t_j]$ as follows,
\begin{align}
u(t_j) = \int_0^{t_j}u_t(z) \mbox{d}z= m^n(t_j) +  Q \zeta(t_j), \quad t_j\in [0, L],\,L>0, \, 0 \leq n \leq N.
\label{eqn:stateconv}
\end{align}

\noindent Let $\mathbf{f}_{1:n}=\left(\mbox{f}_1,\ldots,\mbox{f}_n\right)$ and omit dependence on $\theta$ and $\Psi$, which are fixed.  We wish to update the prior $\left[u \mid \mathbf{f}_{1:n-1} \right]$ given a new model interrogation, $\mbox{f}_n$, to obtain the updated process $\left[u \mid \mathbf{f}_{1:n} \right]$, where $1\leq n \leq N$.

\begin{lem}[Updating]
\label{lem:gaussiansolution}
The stochastic process $\{u(t) \mid \mathbf{f}_{1:n-1}, t\in[0,L]\}$ and its time derivative $\{u_t(t) \mid \mathbf{f}_{1:n-1}, t\in[0,L]\}$ are well-defined and distributed according to a Gaussian probability measure with marginal mean functions, covariance and cross-covariance operators given by,
 \begin{align*}
m^n(t_j)
&=  m^{n-1}(t_j)  +  g_n^{-1}\left\{\mbox{f}_n -  m_t^{n-1}(s_n) \right\}
\smallint_0^{t_j}   C_t^{n-1}(z,s_n) \mbox{d}z,\\
m_t^n(t_j)
&=  m_t^{n-1}(t_j) +  g_n^{-1} \left\{\mbox{f}_n -   m_t^{n-1}(s_n) \right\} \;
C_t^{n-1}(t_j,s_n),\\
 C^n(t_j,t_k)  & =    C^{n-1}(t_j,t_k)  -  g_n^{-1} \smallint_0^{t_j}  C_t^{n-1}(z,s_n)\mbox{d}z \,  \smallint_0^{t_k} C_t^{n-1}(s_n,z)\mbox{d}z ,\\
 C_t^n(t_j,t_k)  & =  C_t^{n-1}(t_j,t_k)  -   g_n^{-1}  C_t^{n-1}(t_j,s_n)\, C_t^{n-1}(s_n,t_k),\\
\smallint_0^{t_j}  C_t^n(z,t_k) \mbox{d}z  & =  \smallint_0^{t_j}C_t^{n-1}(z,t_k) \mbox{d}z -   g_n^{-1}  \smallint_0^{t_j} C_t^{n-1}(z,s_n) \mbox{d}z \, C_t^{n-1}(s_n,t_k),\\
 \smallint_0^{t_k} C_t^n(t_j,z) \mbox{d}z  & = \left\{\smallint_0^{t_k}  C_t^n(z,t_j) \mbox{d}z \right\}^\top
 \end{align*}
\noindent where,
\[ g_n : = \left\{ \begin{array}{ll}
         C_t^{0}(s_1,s_1) & \mbox{if $n = 1$},\\
        C_t^{n-1}(s_n,s_n) + r_{n-1}(s_n) & \mbox{if $n > 1$}, \end{array} \right. \] 
and where $m^0$ and $m_t^0$ are the prior means and $ C^0$ and $ C_t^0$ the prior covariances of the state and derivatives defined in Section \ref{sec:priors}.
\end{lem}

\begin{proof}
We are interested in the conditional distribution of the state $u(t) - m^{n-1}(t)\in\mathcal{F}^*$ and time derivative $u_t(t) - m_t^{n-1}(t) \in\mathcal{F}^*$ given a new model interrogation on a mesh vertex $s_n\in[0, L],L>0$ under the Gaussian error model,
\[\mbox{f}_n - m_t^{n-1}(s_n) =  R \zeta(s_n) + \eta(s_n),\]
where $\eta(s_n)\sim \mathcal{N}(0, \gamma_n)$.

Construct the vector 
\[ \left( u_t - m_t^{n-1},\mbox{f}_n - m_t^{n-1}(s_n) \right)
= \left(  R \zeta,  R \zeta(s_n) + \eta(s_n) \right)
\in\mathcal{F}^*\oplus\, \mathbb{R},\]
where the first element is function-valued and the second element is a scalar.   This vector is jointly Gaussian with mean $M = (0,0)$ and covariance operator $C$ with positive definite cross-covariance operators,
\begin{align}
\begin{array}{ll}
&C_{11} =  R K R ^\dagger 
\quad\quad
C_{12} =   R K R ^\dagger\\
&C_{21} =   R K R ^\dagger \quad\quad
C_{22} =   R K R ^\dagger + \gamma_n.
\end{array}
\label{eqn:cross_covs}
\end{align}

\noindent Since both $\mathcal{F}^*$ and $\mathbb{R}$ are separable Hilbert spaces, it follows from Theorem 6.20 in \cite{Stuart2010} that the random variable $[ u_t - m_t^{n-1} \mid \mbox{f}_n -m_t^{n-1}(s_n) ]$ is well-defined and distributed according to a Gaussian probability measure with mean and covariance,
\begin{align*}
\mbox{E} \left( u_t - m_t^{n-1} \mid \mbox{f}_n - m_t^{n-1}(s_n) \right)
& = C_{12}C_{22}^{-1} (\mbox{f}_n-m_t^{n-1}(s_n)),\\
\mbox{Cov}\left(u_t - m_t^{n-1} \mid \mbox{f}_n - m_t^{n-1}(s_n) \right)
& = C_{11} - C_{12}C_{22}^{-1}C_{21}.
\end{align*}

Similarly, consider the vector 
$\left( u - m^{n-1},\mbox{f}_n - m_t^{n-1}(s_n) \right)
= \left(  Q \zeta,  R \zeta(s_n) + \eta(s_n) \right)
\in\mathcal{F}^*\oplus\, \mathbb{R}$, with mean 
$M = (0,0)$ 
and cross-covariances,
\begin{align*}
\begin{array}{ll}
&C_{11} =  Q K Q ^\dagger 
\quad\quad
C_{12} =   Q K R ^\dagger\\
&C_{21} =   R K Q ^\dagger \quad\quad
C_{22} =   R K R ^\dagger + \gamma_n.
\end{array}
\end{align*}

\noindent By Theorem 6.20 \citep{Stuart2010}, the random variable $[ u - m^{n-1} \mid \mbox{f}_n -m_t^{n-1}(s_n) ]$ is well-defined and distributed according to a Gaussian probability measure with mean and covariance,
\begin{align*}
\mbox{E} \left( u - m^{n-1} \mid \mbox{f}_n - m_t^{n-1}(s_n) \right)
& = C_{12}C_{22}^{-1} (\mbox{f}_n-m_t^{n-1}(s_n)),\\
\mbox{Cov}\left( u - m^{n-1} \mid \mbox{f}_n - m_t^{n-1}(s_n) \right)
& = C_{11} - C_{12}C_{22}^{-1}C_{21}.
\end{align*} 
Cross-covariances are found analogously.

\end{proof}
\vspace{-0.5cm}
The kernel convolution approach to defining the covariances guarantees that the cross-covariance operators,  (\ref{eqn:cross_covs}), between the derivative and $n$ derivative model realizations, are positive definite.

\subsection{Proof of Theorem 3.1}\label{sec:consistency}

Denote the exact solution satisfying (\ref{eqn:IVPsystem}) by $u^*(t)$.  For clarity of exposition, we assume that $u(0) = 0$ and  let $m^0(t)=0$ for $t\in[0, L], L>0$. We define $h =\max_{n=2,\ldots,N}\left(s_{n}-s_{n-1}\right)$ to be the maximum step length between subsequent discretization grid points.  We would like to show that the forward model $\{u(t) \mid N , t\in[0,L] \}$ (sampled via  Algorithm \ref{alg:ivpsolution}) with $r_n(s)=0, s\in[0,L]$ converges in $L^1$ to $u^*(t)$ as $h \to 0$ and $\lambda,\alpha^{-1} = O(h)$, concentrating at a rate proportional to $h$.

Let $\mathbf{f}_{1:n}=\left(\mbox{f}_1,\ldots,\mbox{f}_n\right)$ and omit dependence on $\theta$ and $\Psi$, which are fixed.  For $t\in[0,L]$ we find $n$ such that $t\in[s_{n},s_{n+1}]$, and wish to bound the expected absolute difference $\beta_n(t)$ between the exact solution and a single mixture component of (\ref{eqn:posterior_solution}).  Let $\Phi$ be the standard normal cdf, and $A = \mbox{E}\left( u(t) - u^*(t)\mid \mathbf{f}_{1:n}\right)/\sqrt{ C^n(t,t)}$.
\begin{align}
\beta_n(t) &= \mbox{E}\left( |u(t) - u^*(t) | \mid \mathbf{f}_{1:n} \right) \nonumber\\
 &= \mbox{E}\left( u(t) - u^*(t) \mid \mathbf{f}_{1:n} \right) \left\{ 1 - 2\Phi \left( -A \right)\right\} 
 + \sqrt{\frac{2}{\pi}\, C^n(t,t)} \; \exp \left\{-A^2/2 \right\}, \nonumber\\
& \leq \big| \mbox{E}\left( u(t) - u^*(t) \mid \mathbf{f}_{1:n} \right) \big| + O(h^2)
\label{eqn:foldnorm}
\end{align}
The second equality uses the expectation of a folded normal distribution \citep[][]{LeoneEtAl1961} and Lemma \ref{lem:covariance}.  Next bound the first term of (\ref{eqn:foldnorm}) as follows,
\begin{align}
&
\mbox{E} \left(u(t) - u^*(t) \mid \mathbf{f}_{1:n} \right)\nonumber\\
& = \mbox{E} \left( u(s_n) - u^*(s_n)\mid \mathbf{f}_{1:n} \right) + (t-s_n) \mbox{E} \left( u_t(s_n) - u_t^*(s_n)  \mid \mathbf{f}_{1:n} \right)+ O(h^2)\nonumber\\
& = \mbox{E} \left( u(s_n)  - u^*(s_n) \mid \mathbf{f}_{1:n-1} \right)
 + \frac{ C^{n-1}(s_n,s_n)}{ C_t^{n-1}(s_n,s_n)}
\mbox{E} \left( \mbox{f}_n - u_t(s_n) \mid \mathbf{f}_{1:n-1} \right)  \nonumber\\
& \quad\quad  + (t-s_n) \mbox{E} \left( \mbox{f}_n - u_t^*(s_n) \mid \mathbf{f}_{1:n-1} \right) + O(h^2).
\label{eqn:meandiff}
\end{align}   
The $n$th updated process is mean-square differentiable on [0,L], so a Taylor expansion around $s_n$ gives us the first equality, and the recursion from Lemma \ref{lem:gaussiansolution} gives us the second equality.  Next, we use Jensen's inequality to obtain,
\begin{align*}
\begin{array}{ll}
\beta_n(t) 
 \leq \beta_{n-1}(s_n)  + |t-s_n| \; \mbox{E} \left( |\mbox{f}_n - u_t^*(s_n) | \mid \mathbf{f}_{1:n-1} \right)  + \frac{ C^{n-1}(s_n,s_n)}{ C_t^{n-1}(s_n,s_n)} \; \mbox{E}\left( |\mbox{f}_n - u_t(s_n) | \mid \mathbf{f}_{1:n-1} \right)  + O(h^2).
\end{array}
\end{align*}
 By the Lipschitz continuity of $f$, boundedness of $\mbox{E}\left( |\mbox{f}_n - u_t(s_n)| \mid \mathbf{f}_{1:n-1} \right)$, and recursive construction of the covariance, we obtain,
\begin{align*}
\begin{array}{ll}
\beta_n(t) & \leq \beta_{n-1}(s_n)  + L|t-s_n| \beta_{n-1}(s_n)  
 + O(h^2), \\
& = \beta_{n-1}(s_n)\left(1  + L|t-s_n|\right)  
+ O(h^2).
\end{array}
\end{align*}

\noindent It can be shown 
\citep[e.g.][p.67-68]{Butcher2008} that the following inequality holds:
\begin{align*}
\beta_n(t)
& \leq  
 \left\{ 
  \begin{array}{l l}
   \beta_0(s_1) + hB(t-a), & \quad L = 0,   \\
  \exp \{ (t-a)L\} \beta_0(s_1)  + \exp \{ (t-a)L-1\} hB/L, & \quad L > 0,   \\
  \end{array} \right.
\end{align*}
where $B$ is the constant upper bound on all the remainders.  This expression tends to $0$ as $\alpha^{-1},\lambda,h\to 0$, since $ \beta_0(s_1) =0$.  Then, taking the expectation of $\beta_n(t)$ with respect to $\mathbf{f}=(\mbox{f}_1,\ldots,\mbox{f}_N)$, we obtain,
\begin{align*}
\mbox{E} \left( |u(t) - u^*(t)|  \mid  N \right)
= O(h), \quad \mbox{as }\alpha^{-1},\lambda\to 0.
\end{align*}
Thus, the probabilistic solution (\ref{eqn:posterior_solution}) converges in $L^1$ to $u^*(t)$ at the rate $O(h)$.  Note that the assumption that auxiliary parameters, $\lambda$ and $\alpha^{-1}$, associated with the solver tend to zero with the step size is analogous to maintaining a constant number of steps in a $k$-step numerical method regardless of the step size.

\subsection{Properties of the covariance}\label{sec:covarianceproperties}

In this section we present some results regarding covariance structures that are used in the proof of Theorem \ref{thm:meanconvergence}.

\begin{lem}
\label{lem:covinequality}
For $1<n\leq N$ and $t\in[0,L], L>0$, the variances for the state and derivative obtained sequentially via Algorithm \ref{alg:ivpsolution} satisfy:
\begin{align*}
C^n(t,t)  & \leq  C^{1}(t,t), \\
 C_t^n(t,t)  & \leq  C_t^1(t,t). 
\end{align*}
\end{lem}
\begin{proof}
We use the fact that $ C_t^n(t,t)\geq 0$ for all $n$ and the recurrence from Lemma \ref{lem:gaussiansolution}, we obtain:
\begin{align*}
 C_t^n(t,t)  & =  C_t^{n-1}(t,t)  -  g_n^{-1}  C_t^{n-1}(t,s_n) \,  C_t^{n-1}(s_n,t) ,\\
& \leq   C_t^{n-1}(t,t) \leq \cdots \leq  C_t^{1}(t,t).
\end{align*}
Similarly,
\begin{align*}
 C^n(t,t)  & =  C^{n-1}(t,t)  -  g_n^{-1} \smallint_0^{t}  C_t^{n-1}(z,s_n)\mbox{d}z \,  \smallint_0^{t} C_t^{n-1}(s_n,z)\mbox{d}z ,\\
& \leq   C^{n-1}(t,t) \leq \cdots \leq  C^{1}(t,t).
\end{align*}
\end{proof}
\begin{lem}
\label{lem:covariance}
The covariances, $ C^n(t_j,t_k)$ and $ C_t^n(t_j,t_k)$, obtained at each step of Algorithm \ref{alg:ivpsolution} tend to zero at the rate $O(h^4)$, as $h\to 0$ and $\lambda, \alpha^{-1} = O(h)$ if the covariance function $ R_\lambda$ is stationary and satisfies:
\begin{align}
 RR (t,t) - \frac{RR (t+d,t) RR (t+d,t)}{RR (t,t)} &= O(h^4),  \quad  \lambda, \; \alpha^{-1} = O(h), \; h\to 0, \label{eqn:as1}\\
 QQ (t,t) -  \frac{QR (t+d,t) QR ^\dagger(t+d,t)}{RR (t,t)} & = O(h^4),  \quad  \lambda, \; \alpha^{-1} = O(h), \;h\to 0,
\label{eqn:as2}
\end{align} 
\noindent where $d>0$, $t,t+d \in [0,L]$, and $\lambda\geq h$.
\end{lem}
\begin{proof} 
From Lemma \ref{lem:covinequality} and assumption (\ref{eqn:as1}) we obtain,
\begin{align*}
 C_t^{n}(t,t) \leq  C_t^{1}(t,t) =  RR (t,t) -  \frac{RR (t,s_1) RR (s_1,t)}{RR (s_1,s_1)} = O(h^4),  \quad  \lambda, \alpha^{-1} = O(h),\; h\to 0, \; t\in[0,L], \; 1\leq n \leq N. 
\end{align*}
\noindent Similarly, using Lemma \ref{lem:covinequality}  and assumption (\ref{eqn:as2}) yields,
\begin{align*}
 C^{n}(t,t) \leq  C^{1}(t,t) =  QQ (t,t) -  \frac{QR (t,s_1) QR ^\dagger(t,s_1)}{RR (s_1,s_1)} = O(h^4),  \quad  \lambda, \alpha^{-1} = O(h),\; h\to 0, \; t\in[0,L], \; 1\leq n \leq N. 
\end{align*}
Then, by the Cauchy\--Schwarz inequality,
\[
\big|  C_t^n(t_j,t_k) \big| , \big|  C^n(t_j,t_k) \big| = O(h^4),  \;  \lambda, \alpha^{-1} = O(h),\; h\to 0, \;  t_j, t_k\in[0,L], \; 1\leq n \leq N. 
\]
\end{proof}
\vspace{-0.5cm}
The square exponential and uniform covariance functions considered in this paper are stationary and symmetric and satisfy conditions (\ref{eqn:as1}) and (\ref{eqn:as2}).  
\begin{lem}
\label{lem:Lambda}
The covariance $C_t^{n-1}(s_n,s_n)$, $1\leq n \leq N$ in Algorithm \ref{alg:ivpsolution},
tends to zero at the rate $O(h^4)$ as $N\to\infty$ and $h\to 0$ when conditions (\ref{eqn:as1}) and (\ref{eqn:as2}) are satisfied.  
\end{lem}
\vspace{-2cm}
\begin{proof}
The proof follows immediately from Lemma \ref{lem:covariance}.
\end{proof}

\subsection{Some covariance kernels and their convolutions}\label{sec:exactconv}

 The examples in this paper utilize two types of covariance functions, although the results presented are more generally applicable.  Imposing unrealistically strict smoothness assumptions on the state space by choice of covariance structure may introduce estimation bias if the exact solution is less smooth than expected.  In this paper we work with stationary derivative covariance structures for simplicity,  however there will be classes of problems where non-stationary kernels may be more appropriate. 

The infinitely differentiable squared exponential covariance is based on the kernel
$R_\lambda(t_j,t_k) = \exp \left\{ -{(t_j-t_k)^2}/{2\lambda^2}\right\}$.  We utilize this covariance structure in the forward models of the Lorenz63 system, the Navier-Stokes equations, the Lane-Emden mixed boundary value problem, and the heat equation.  In contrast, solutions for systems of delay initial function problems are often characterized by second derivative discontinuities at the lag locations.  For these forward models we utilize the uniform derivative covariance structure, obtained by convolving the kernel $R_\lambda(t_j,t_k) =  \mbox{1}_{(t_j-\lambda,t_j+\lambda)}(t_k)$ with itself, is non-differentiable.  Given the shape of the uniform kernel, choosing the length-scale $\lambda$, to be greater than one half of the maximum step length ensures that the resulting step ahead prediction captures information from at least one previous model interrogation.

Closed form expressions for the pairwise convolutions for the two covariance functions are provided below and implemented in the accompanying software.  Let $R_\lambda$ be the squared exponential kernel and let $Q_\lambda$ be its integrated version. Then,
\begin{align*}
\begin{array}{ll}
\alpha RR (t_j,t_k) &= 
\sqrt{\pi}\lambda \exp\left \{ -\tfrac{(t_j - t_k)^2}{4\lambda^2}\right \},
\\
\alpha QR (t_j,t_k) &=
\pi\lambda^2 \left( \mbox{erf}\left\{\tfrac{t_j - t_k}{2\lambda}\right\}
+ \mbox{erf}\left\{\tfrac{t_k-a}{2\lambda}\right\} \right),
\\
\alpha QQ (t_j,t_k) 
&=
\pi\lambda^2 \left( (t_j-a)\mbox{erf}\left\{\tfrac{t_j-a}{2\lambda}\right\}
- (t_k-t_j)\mbox{erf}\left\{\tfrac{t_k-t_j}{2\lambda}\right\}
+ (t_k-a)\mbox{erf}\left\{\tfrac{t_k-a}{2\lambda}\right\}
\right),
\\
&+ \, 2\sqrt{\pi}\lambda^3 \left( \exp\left\{-\tfrac{(t_j-a)^2}{4\lambda^2}\right\}
-  \exp\left\{-\tfrac{(t_k-t_j)^2}{4\lambda^2}\right\}
+  \exp\left\{-\tfrac{(t_k-a)^2}{4\lambda^2}\right\}
 - 1
 \right).
 \end{array}
\end{align*}
 Next, let  $R_\lambda$ be the uniform kernel and let $Q_\lambda$ be its integrated version. Then,
\begin{align*}
\begin{array}{ll}
\alpha RR (t_j,t_k) &=\{ \mbox{min}(t_j,t_k) - \mbox{max}(t_j,t_k) +2\lambda \} \; \mbox{1}_{(0,\infty)}\{ \mbox{min}(t_j,t_k)-\mbox{max}(t_j,t_k) +  2\lambda \} ,\\
\alpha QR (t_j,t_k) 
&=  2\lambda \{\mbox{min}(t_j-\lambda,t_k+\lambda) - \mbox{max}(a+\lambda,t_k-\lambda)\} \\ 
&  \quad\mbox{1}_{(0,\infty)}\{ \mbox{min}(t_j-\lambda,t_k+\lambda) - \mbox{max}(a+\lambda,t_k-\lambda) \}\\
& + (t_j-a) \{\mbox{min}(t_j,t_k)-a-2\lambda\} \,\mbox{1}_{(0,\infty)}\{ \mbox{min}(t_j,t_k) - a - 2\lambda\}\\
& + \left[(t_j+\lambda)\{\mbox{min}(t_j,t_k)- \mbox{max}(a+\lambda,t_j-\lambda,t_k-\lambda) +\lambda  \} \right. \\
& \left. \quad - \tfrac{1}{2} \mbox{min}(t_j+\lambda,t_k+\lambda)^2 + \tfrac{1}{2}\mbox{max}(a+\lambda,t_j-\lambda,t_k-\lambda)^2  \right]\\
&   \quad  \quad\mbox{1}_{(0,\infty)}\{\mbox{min}(t_j,t_k)+\lambda - \mbox{max}(a+\lambda,t_j-\lambda,t_k-\lambda) \}\\
&+\left[(\lambda-a)\{\mbox{min}(a+\lambda,t_j-\lambda,t_k+\lambda) - (t_k-\lambda) \} \right.\\
& \left.\quad+ \tfrac{1}{2}\mbox{min}(a+\lambda,t_j-\lambda,t_k+\lambda)^2- \tfrac{1}{2}(t_k-\lambda)^2\right] \\
&\quad\quad \mbox{1}_{(0,\infty)}\{\mbox{min}(a+\lambda,t_j-\lambda,t_k+\lambda) - (t_k-\lambda) \}, \\
\alpha QQ (t_j,t_k) 
&= 4\lambda^2 \{\mbox{min}(t_j,t_k) - a -2\lambda\}\, \mbox{1}_{(0,\infty)}\{\mbox{min}(t_j,t_k) - a - 2\lambda\}\\
&  +2\lambda\left[(t_k+\lambda)\{ \mbox{min}(t_j-\lambda,t_k+\lambda) - \mbox{max}(a+\lambda,t_k-\lambda)\} \right.  \\
& \left. \quad - \frac{1}{2}\mbox{min}(t_j-\lambda,t_k+\lambda)^2+ \tfrac{1}{2}\mbox{max}(a+\lambda,t_k-\lambda)^2 \right] \\
&\quad\quad \mbox{1}_{(0,\infty)}\{\mbox{min}(t_j-\lambda,t_k+\lambda) - \mbox{max}(a+\lambda,t_k-\lambda)\}\\
&  +\left[ \tfrac{1}{3} \mbox{min}(a+\lambda,t_j-\lambda,t_k-\lambda) ^3 - \tfrac{1}{3}(a-\lambda)^3 \right.\\ 
&\left. \quad + (\lambda-a)\{ \mbox{min}(a+\lambda,t_j-\lambda,t_k-\lambda)^2 - (a-\lambda)^2\} \right.\\
& \left. \quad +(\lambda-a)^2\{ \mbox{min}(a+\lambda,t_j-\lambda,t_k-\lambda) - (a-\lambda)\}\right]\\
&\quad\quad \mbox{1}_{(0,\infty)}\{ \mbox{min}(a+\lambda,t_j-\lambda,t_k-\lambda) - (a-\lambda)\}\\
&    + (t_j-a)\left[\tfrac{1}{2}\mbox{min}(a+\lambda,t_k-\lambda) ^2- \tfrac{1}{2}(t_j-\lambda)^2 \right.\\
&\left.\quad+ (\lambda-a)\{\mbox{min}(a+\lambda,t_k-\lambda) - (t_j-\lambda) \}\right]\\
&\quad\quad \mbox{1}_{(0,\infty)}\{\mbox{min}(a+\lambda,t_k-\lambda) - (t_j-\lambda) \}\\
& +(t_k-a)\left[ \tfrac{1}{2}\mbox{min}(a+\lambda,t_j-\lambda) ^2- \tfrac{1}{2}(t_k-\lambda)^2\right.\\
&\left.\quad + (\lambda-a)\{\mbox{min}(a+\lambda,t_j-\lambda) - (t_k-\lambda)\}\right] \, \mbox{1}_{(0,\infty)}\{\mbox{min}(a+\lambda,t_j-\lambda) - (t_k-\lambda)\}\\
&  +2\lambda\left[(t_j+\lambda)\{\mbox{min}(t_j+\lambda,t_k-\lambda) - \mbox{max}(a+\lambda,t_j-\lambda)\} \right.\\
&\left.\quad - \tfrac{1}{2}\mbox{min}(t_j+\lambda,t_k-\lambda)^2  +  \tfrac{1}{2}\mbox{max}(a+\lambda,t_j-\lambda)^2 \right] \\
&\quad\quad \mbox{1}_{(0,\infty)}\{\mbox{min}(t_j+\lambda,t_k-\lambda) - \mbox{max}(a+\lambda,t_j-\lambda)\}\\
&  +\left[(t_j+\lambda)(t_k+\lambda)\{\mbox{min}(t_j,t_k)+\lambda - \mbox{max}(a+\lambda,t_j-\lambda,t_k-\lambda)\} \right.\\
&\left.\quad- \tfrac{1}{2}(t_j+t_k+2\lambda)\{\mbox{min}(t_j+\lambda ,t_k+\lambda )^2- \mbox{max}(a+\lambda,t_j-\lambda,t_k-\lambda)^2\} \right.\\
&\left.\quad  + \tfrac{1}{3}\mbox{min}(t_j+\lambda,t_k+\lambda)^3 - \tfrac{1}{3}\mbox{max}(a+\lambda,t_j-\lambda,t_k-\lambda)^3 \right] \\
&\quad\quad \mbox{1}_{(0,\infty)}\{\mbox{min}(t_j,t_k)+\lambda - \mbox{max}(a+\lambda,t_j-\lambda,t_k-\lambda)\}\\
& + (t_j-a)(t_k-a) \{ a+2\lambda - \mbox{max}(t_j,t_k) \}\, \mbox{1}_{(0,\infty)}\{a+2\lambda - \mbox{max}(t_j,t_k)\}.
\end{array}
\end{align*}
%

\bibliographystyle{ba}
\bibliography{refs-sup}




\end{document}